%% file: main.tex
\documentclass{llncs}
\pdfoutput=1

\input{packages}
\input{macros}

\begin{document}

\title{Multiparty Session Types with a Bang!}

\author{
  Matthew~Alan {Le Brun} \orcidlink{0000-0001-7394-0122} \and 
  Simon Fowler \orcidlink{0000-0001-5143-5475} \and
  Ornela Dardha \orcidlink{0000-0001-9927-7875}
}

\institute{
  University of Glasgow \\ \email{m.le-brun.1@research.gla.ac.uk} \\ \email{simon.fowler@glasgow.ac.uk} \\ \email{ornela.dardha@glasgow.ac.uk}
}

\maketitle

\begin{abstract}
  Replication is an alternative construct to recursion for describing infinite behaviours in the $\pi$-calculus.
  In this paper we explore the implications of including type-level replication
  in \emph{Multiparty Session Types (MPST)}, a behavioural type theory for
  message-passing programs.
  We introduce $\MPSTb$, a session-typed multiparty process calculus with
  replication and first-class roles.
  %
  We show that replication is \emph{not} an equivalent alternative to recursion
  in MPST, and that using both replication \emph{and} recursion in one type system in
  fact allows us to express both context-free protocols and protocols that
  support mutual exclusion and races. We demonstrate the expressiveness of
  $\MPSTb$ on examples including binary tree serialisation, dining
  philosophers, and a model of an auction, and explore the implications
  of replication on the decidability of typechecking.
\end{abstract}



\input{intro}

\input{lang}
\input{types}
\input{metatheory}

\input{discussion}

\input{related}

\input{conc}

\paragraph{Acknowledgements.} 
We thank the anonymous reviewers for their detailed and insightful comments.
This work was supported by EPSRC Grants EP/X027309/1 (Uni-pi) and  EP/T014628/1
(STARDUST).

\clearpage
\bibliographystyle{splncs04}
\bibliography{my}

\clearpage
\changetext{}{10em}{-5em}{-5em}{}
\appendix

\input{cong}
\input{prelims}
\input{sr}
\input{decidability}

\end{document}

%% file: packages.tex
\usepackage{bm}

\usepackage[inline]{enumitem}

\usepackage{microtype}

\usepackage{mathpartir}

\usepackage{centernot}

\usepackage{orcidlink}

\usepackage{appendix}

\usepackage[strict]{changepage}

\usepackage{amsfonts}

\usepackage{amsmath}

\usepackage{amssymb}

\usepackage{framed}

\usepackage{centernot}

\usepackage{xspace} 

\usepackage{cleveref}

%% file: macros.tex
\definecolor{colabred}{RGB}{242, 60, 60}
\definecolor{context}{RGB}{0, 51, 102}
\definecolor{highlight}{rgb}{0.97,0.91,0.81}

\definecolor[named]{ACMBlue}{cmyk}{1,0.1,0,0.1}
\definecolor[named]{ACMYellow}{cmyk}{0,0.16,1,0}
\definecolor[named]{ACMOrange}{cmyk}{0,0.42,1,0.01}
\definecolor[named]{ACMRed}{cmyk}{0,0.90,0.86,0}
\definecolor[named]{ACMLightBlue}{cmyk}{0.49,0.01,0,0}
\definecolor[named]{ACMGreen}{cmyk}{0.20,0,1,0.19}
\definecolor[named]{ACMPurple}{cmyk}{0.55,1,0,0.15}
\definecolor[named]{ACMDarkBlue}{cmyk}{1,0.58,0,0.21}

\newcommand{\type}[1]{{\color{ACMDarkBlue} #1}}

\crefname{equation}{}{}
\crefname{rq}{research question}{research questions}
\crefname{l}{line}{lines}
\crefname{rule}{rule}{rules}
\crefname{figure}{figure}{figures}

\newcounter{rule}
\crefformat{rule}{#2[#1]#3}
\Crefformat{rule}{#2[#1]#3}

\renewcommand{\infer}[5][]{%
    \def\therule{#2}%
    \refstepcounter{rule}%
    \label{rule:#3}%
    \inferrule*[lab={#2}, right={#1}, vcenter]{#4}{#5}%
}

\newtheorem{innercustomthm}{Theorem}
\newenvironment{customthm}[1]
  {\renewcommand\theinnercustomthm{#1}\innercustomthm}
  {\endinnercustomthm}

\newcommand*{\eg}{\textit{e.g.}}

\newcommand*{\ie}{\textit{i.e.,}}

\newcommand*{\etc}{\textit{etc.}}

\newcommand*{\cf}{\textit{cf.}}

\newcommand{\dom}{\textsf{dom}}

\newcommand{\pEnd}{\textsf{end}}

\newcommand{\extract}{\textsf{assoc}}
\newcommand{\sizeof}[1]{|#1|}

\newcommand{\midspace}{\;\mid\;}
\newcommand{\tab}{\ \ \ \ }

\let\oldAnd\&
\renewcommand{\&}{\oldAnd}
\renewcommand{\|}{\, |\,}
\renewcommand{\mid}{\ \bigm\vert\ }
\renewcommand{\o}{\oplus}

\newcommand{\ses}[1]{{\color{ACMBlue} #1}}

\newcommand{\poly}[1]{\widetilde{#1}}

\newcommand{\rc}{\mathbb{C}}

\newcommand{\0}{\bm 0}
\newcommand{\msgLabel}[1]{\textsf{#1}}
\newcommand{\recVar}[1]{\textnormal{\textrm{#1}}}
\newcommand{\defin}[3]{\textsf{def } \recVar{#1}(#2) = #3 \textsf{in }}
\newcommand{\definD}[1]{\textsf{def } #1 \textsf{in }}
\newcommand{\then}{\ .\ }

\newcommand{\role}[1]{{\color{ACMRed} {\ensuremath{\bm {#1}}}}}

\newcommand{\new}[2]{(\nu \ses{#1} : #2)\ }
\newcommand{\newnt}[1]{(\nu \ses{#1})}

\newcommand{\rcv}[3]{ #1[#2] \; {\&}\;  #3 \then }
\newcommand{\rcvI}[3]{ #1[#2] \; {\&}_{i\in I}\;  #3 \then }

\newcommand{\snd}[3]{#1[#2] \; {\o}\;  #3 \then }

\newcommand{\intT}{\type{\textsf{int}}}

\newcommand{\tEnd}{\type{\textsf{\textbf{end}}}}

\let\oldTheta\Theta
\renewcommand{\Theta}{\type{\oldTheta}}

\let\oldGamma\Gamma
\renewcommand{\Gamma}{\type{\oldGamma}}
\newcommand{\Gammab}[1]{\type{\Gamma_{#1}}}

\let\oldDelta\Delta
\renewcommand{\Delta}{\type{\oldDelta}}

\newcommand{\cons}{\,;}
\newcommand{\cSplit}{\,\type{\cdot}\,}

 \newcommand{\insnew}[1]{\hookleftarrow #1}

\let\oldPsi\Psi
\renewcommand{\Psi}{\type{\oldPsi}}

\newcommand{\fv}{\textsf{fv}}
\newcommand{\fpv}{\textsf{fpv}}
\newcommand{\frv}{\textsf{frv}}
\newcommand{\fs}{\textsf{fs}}

\newcommand{\dpv}{\textsf{dpv}}

\newcommand{\redP}{\rightarrow}
\newcommand{\subT}{\type{\leqslant}}
\newcommand{\supT}{\geqslant}
\newcommand{\hasT}{\,\type{:}\,}
\newcommand{\hasTT}{{\type{:}}}
\newcommand{\redCstar}{\type{\rightarrow^*}}
\newcommand{\redS}[1]{\;\type{\xrightarrow{#1}}\;}
\newcommand{\redC}{\type{\rightarrow}}

\newcommand{\beh}{\textsf{beh}}
\newcommand{\unf}{\textsf{unf}}
\newcommand{\tf}{\textsf{tf}}
\newcommand{\lf}{\textsf{lf}}

\newcommand{\act}{\type{\gamma}}
\newcommand{\aIn}[5]{\ses{#1}\type{{:}#2\oldAnd#3{:}#4(#5)}}
\newcommand{\aOut}[5]{\ses{#1}\type{{:}#2{\o}#3{:}#4(#5)}}
\newcommand{\aCom}[4]{\ses{#1}\type{{:}#2{,}#3{:}#4}}

\newcommand{\p}{\role{p}}
\newcommand{\q}{\role{q}}

\newcommand{\m}{\msgLabel{m}}
\newcommand{\ok}{\msgLabel{ok}}

\newcommand{\ping}{\msgLabel{ping}}
\newcommand{\pong}{\msgLabel{pong}}

\newcommand{\tree}{\msgLabel{tree}}
\newcommand{\node}{\msgLabel{node}}
\newcommand{\leaf}{\msgLabel{leaf}}
\newcommand{\skipy}{\msgLabel{skip}}

\newcommand{\lock}{\msgLabel{lock}}
\newcommand{\acquire}{\msgLabel{acquire}}
\newcommand{\release}{\msgLabel{release}}
\newcommand{\done}{\msgLabel{done}}

\newcommand{\take}{\msgLabel{take?}}
\newcommand{\give}{\msgLabel{give}}
\newcommand{\fin}{\msgLabel{fin}}

\newcommand{\bid}{\msgLabel{bid}}

\newcommand{\notavail}{\msgLabel{not-avail}}
\newcommand{\yes}{\msgLabel{yes}}
\newcommand{\no}{\msgLabel{no}}

\newcommand{\St}{\type{S}}

\newcommand{\Gammap}{\type{\Gamma^\prime}}
\newcommand{\Gammapp}{\type{\Gamma^{\prime\prime}}}

\newcommand{\sub}[2]{\{#1/#2\}}

\newcommand{\hl}[1]{\colorbox{highlight}{#1}}
\newcommand{\reason}[2]{\textnormal{(#1)}\label{l:#2}}

\newcommand{\MPSTb}{\textsf{MPST!}\xspace}
\newcommand{\MAGpi}{MAG$\pi$\xspace}
\newcommand{\MAGpib}{MAG$\pi$!\xspace}

\newcommand{\smallminus}{\scalebox{0.4}[1.0]{\( - \)}}
\newcommand{\smallequals}{\scalebox{0.6}[1.0]{\( = \)}}
\newcommand{\smallplus}{\scalebox{0.5}{\( + \)}}

\newcommand{\header}[1]{
  \begin{flushleft}
    \textbf{#1}
  \end{flushleft}
}
\newcommand{\headersig}[2]{
  \begin{flushleft}
    \textbf{#1} \hfill \framebox{#2}
  \end{flushleft}
}

\newcommand{\headertwo}[2]{
  \begin{flushleft}
    \textbf{#1} \hfill #2
  \end{flushleft}
}

%% file: intro.tex
\section{Introduction}\label{sec:intro}

Our world is powered by a multitude of computer systems working together by
\emph{communicating}, \ie\ sending and receiving messages according to some
\emph{protocol}. It is therefore vital to \emph{verify} the correctness of both
communication protocols and their implementations, to ensure our programs behave
according to their specifications, and to guarantee that specifications are
indeed \emph{safe}.

\emph{Session
types}~\cite{DBLP:journals/iandc/DardhaGS17,DBLP:conf/concur/Honda93,DBLP:conf/esop/HondaVK98,DBLP:journals/iandc/Vasconcelos12}
provide a lightweight method by which a developer can ensure safety of, and
conformance to, communication protocols.
Session types can be thought of as \emph{types for protocols} which can be
attached to a communication \emph{channel} to specify \emph{how} it should be
used, and can be used to detect issues such as \emph{communication mismatches}
and \emph{deadlocks} early in the development process.
\emph{Multiparty session types}
(MPST)~\cite{DBLP:conf/sfm/CoppoDPY15,DBLP:journals/jacm/HondaYC16,DBLP:journals/darts/ScalasDHY17}
generalise binary session types to allow reasoning about communication between
two \emph{or more} participants, and have been shown to be expressive
enough to capture a range of practical protocols such as the OAuth~2
authentication protocol~\cite{DBLP:journals/pacmpl/ScalasY19}.

\begin{example}[Client-Server-Worker]\label{ex:intro}
    Using generalised MPST~\cite{DBLP:journals/pacmpl/ScalasY19}, we describe the types for three participants in a simple work-offloading system.
    {\small\begin{equation}
        \begin{array}{c}
            \type{S_\role{c}} := \type{\role s {\o} \msgLabel{req}(\texttt{int}) \then \role w \& \msgLabel{ans}(\texttt{str})}
            \\[0.5em]
            \begin{array}{c c}
                \type{S_\role{s}} := \type{\role c {\&} \msgLabel{req}(\texttt{int}) \then \role w {\o} \msgLabel{fw}(\texttt{int})}\tab
                &\tab
                \type{S_\role{w}} := \type{\role s {\&} \msgLabel{fw}(\texttt{int}) \then \role c {\o} \msgLabel{ans}(\texttt{str})}
            \end{array}
        \end{array}
    \end{equation}}%
    The above describes types for a $\role c$lient, $\role s$erver, and $\role w$orker respectively.
    The $\role c$lient, having type $\type{S_\role{c}}$, sends ($\type{\o}$) a \type{\msgLabel{req}}uest to the $\role s$erver with payload type \type{\texttt{int}}, then (\type{.}) waits to receive ($\type{\&}$) an \type{\msgLabel{ans}}wer from the $\role w$orker with payload type \type{\texttt{str}}.
    The $\role s$erver, upon receiving the \type{\msgLabel{req}}uest from the $\role c$lient, \type{\msgLabel{f}}or\type{\msgLabel{w}}ards it to the $\role w$orker.
    Lastly the $\role w$orker, after receiving the \type{\msgLabel{f}}or\type{\msgLabel{w}}arded request from the $\role s$erver, sends the \type{\msgLabel{ans}}wer to the $\role c$lient.
    A MPST system verifies that any written program code conforms to this
    specification (known as \emph{session fidelity}), and that this protocol is
    \emph{safe}---\ie\ that processes send and receive messages of compatible
    types.
\end{example}

Despite the potential of MPST for safe distributed programming, there remain
limitations to the theory that impede their adoption for practical systems.
For instance, generalising \Cref{ex:intro} to multiple workers in the style 
of a load-balancer is non-trivial and has inspired a series of work on the 
generalisation of \emph{direction of choice}~\cite{DBLP:phd/dnb/Stutz24}.
Further, generalising the number of \emph{clients} is also non-trivial---typically, MPST theories assume a global view of \emph{all} participants in a session. 
Lastly, the \emph{objects} of actions in MPST (e.g., the recipient of a sent 
message) are \emph{hard coded} because role names are \emph{constants}.
%


\begin{example}[Load Balancer for $n$ clients]\label{ex:lb-intro}
    We introduce \emph{replication} and \emph{first-class roles}, combined with undirected choice in sends, to generalise \Cref{ex:intro} to support \emph{two workers} and \emph{any number} of clients.
    {\small\begin{equation}
        \begin{array}{r l}
            \type{S_\role{s}} &:= \type{!\role \alpha {\&} \msgLabel{req}(\texttt{int}) \then {\o} \left\{\begin{array}{l}
                \role{w_1} \msgLabel{fw}(\texttt{int}, \role \alpha) \then \role \alpha {\o} \msgLabel{wrk}(\role{w_1}) \\
                \role{w_2} \msgLabel{fw}(\texttt{int}, \role \alpha) \then \role \alpha {\o} \msgLabel{wrk}(\role{w_2})
            \end{array}\right.
            }
            \\[1em]
            \type{S_\role{w_i}} &:= \type{!\role s {\&} \msgLabel{fw}(\texttt{int},\role \gamma) \then \role \gamma {\o} \msgLabel{ans}(\texttt{str})} \tab\tab \textnormal{for } i \in \{1,2\}
        \end{array}
    \end{equation}}%
    The first difference is the use of the \emph{bang} (\type{!}) operator to denote a \emph{replicated} action, \ie\ one which may occur \emph{any number} of times.
    This makes the server agnostic to the \emph{number} of requests.
    Second, a server now waits for requests---not from a specific client---but from \emph{any} participant, binding the name of the sender to role variable $\role \alpha$.
    This makes the server agnostic to the \emph{source} of the requests.
    The server then makes a \emph{choice} to forward the request to one of two workers, notably whilst passing the name of the client as one of the payloads.
    Finally, the server informs the client of the choice it made by sending the name of the worker in that branch.
    The worker type is also updated to be replicated, as it is dependent on the number of requests forwarded by the server.
    Notably, it receives the name of the client in the forward message, binding it to $\role \gamma$, and uses it to send the final answer. 
    A client may now be defined as:
    {\small\begin{equation}
        \type{S_\role{c_j}} := \type{\role s {\o} \msgLabel{req}(\texttt{int}) \then \role s {\&} \msgLabel{wrk}(\role \omega) \then \role \omega \& \msgLabel{ans}(\texttt{str})} \tab \textnormal{for any } j \in \mathbb{N}
    \end{equation}}%
    As a result of the \emph{replicated types} and \emph{first-class role names} 
    on the server-side, we may instantiate \emph{any number of clients} and 
    have them make \emph{any number of requests}---all without changing the server-side protocol.
    Conversely, updating the number of workers has \emph{no impact} on the types of clients.
    Thus, this extension promotes \emph{modular} design of components in multiparty systems.
\end{example}

In fact, as we will see in~\Cref{sec:disc}, the addition of
replication---especially when it is used in tandem with recursion---has several
surprising consequences, in particular allowing us to describe \emph{context-free
protocols} as well as protocols that deal with \emph{races} and \emph{mutual exclusion}.


\subsubsection{Contributions.}
The overarching contribution of this paper is the first integration of
replication and first-class roles into a generalised MPST calculus, and
an exploration of the impacts of these extensions on expressiveness and decidability.
%
Our specific contributions are as follows:
\begin{enumerate}
    \item We present \MPSTb, the first multiparty
        session-typed language with \emph{replication} and \emph{first-class
        roles} (\Cref{sec:lang}), and prove its \emph{metatheory} in the form of \emph{subject
        reduction} and \emph{session fidelity} properties (\Cref{sec:meta}).
    \item We show several expressiveness results through a series of
        representative examples (\Cref{sec:exp}): in particular,
        replication lifts the expressive power of types and thus we give the
        first account of \emph{context-free} MPST. We show that combining
        both replication and recursion allows us to model races and
        mutual exclusion; we demonstrate nontrivial examples
        including \emph{binary tree serialisation}, the \emph{dining
        philosophers problem}, and an \emph{auction service}.
    \item We demonstrate the impacts of replication on the decidability of
        typechecking~(\Cref{sec:decidability}). We show that the decidability of
        typechecking is contingent on the decidability of a given safety
        property, and demonstrate conditions guaranteeing a property to be
        decidable.  Finally we show two syntactic approximations to allow us to
        verify that a property is decidable.
\end{enumerate}

\Cref{sec:related} gives an account of related work and \Cref{sec:conc}
concludes.

%% file: lang.tex
\section{Multiparty Session Types with a Bang!}\label{sec:lang}
In this section we introduce $\MPSTb$, a conservative extension of existing
multiparty session
calculi~\cite{DBLP:conf/sfm/CoppoDPY15,DBLP:journals/darts/ScalasDHY17,DBLP:journals/pacmpl/ScalasY19}
with support for \emph{replication} and \emph{first-class roles}.

\begin{figure}[t]
    {\small\begin{align*}
        c &::= \ses{s}[\q] \mid x  &(\textit{session w/ role, channel variable})\\
        \role{\rho} &::= \q \mid \role \alpha &(\textit{role name value, \hl{role name variable}})\\
        a &::= c \mid \role{\alpha} \tab V ::= c \mid \role{\rho}  &(\textit{names, values}) \\
        b &::= x \mid \role \alpha \tab d ::= \ses{s}[\q] \mid \q &(\textit{binders, concrete values}) \\
        P,Q &::= P \| Q \mid \newnt{s}\, P \mid \0 &(\textit{composition, restriction, termination})\\
            &\mid {\textstyle\sum_{i \in I}} \snd{c_i}{\role{\rho_i}}{\msgLabel{m}_i\langle \poly{V_i} \rangle} P_i &(\textit{choice of sends})\\
            &\mid [{!}]\,c[\role{\rho}] \; \&_{i \in  I}\; \m_i (\poly{b_i}) \then P_i &(\textnormal{\hl{[\emph{replicated}]} \emph{branching receive}})\\
            &\mid \definD{D\;} Q \mid \recVar{X}\langle \poly{c} \rangle &(\textit{process definition, process call})\\
        D &::= \recVar{X}(\poly{x}) = P  &(\textit{process declaration}) 
    \end{align*}}
    \caption{Syntax of \MPSTb}
    \label{fig:lang-syntax}
\end{figure}

\subsection{Language}

\Cref{fig:lang-syntax} shows the syntax of $\MPSTb$.

\paragraph{Names, values, and binders.}
A \emph{session name}, ranged over by $\ses{s}, \ses{s^\prime}, \dots$,
represents a collection of interconnected participants.  A \emph{role} is a
participant in a multiparty communication protocol, and each \emph{communication
endpoint} $\ses{s}[\q]$ is obtained by indexing a session name with a role.
In contrast to existing MPST calculi, \MPSTb supports \emph{first-class} roles,
meaning that a role may be communicated as part of a message. To this end, a role
$\role{\rho}$ may either be a concrete role \role{p} (e.g., \role{s} or \role{w}
in our load balancing example) or a \emph{role variable} \role{\alpha}.
A name $a$ is either an endpoint, a variable, or a role variable, whereas
a value $V$ is either a channel or a role.
Binders $b$ are used when receiving a message and can either be a variable
binder or a role variable binder.
Concrete values $d$ are used when sending a message (at runtime) and are either
an endpoint or a role name value.

\paragraph{Processes.}
Processes are ranged over by $P,Q,R,\dots$: process $P \| Q$ denotes
$P$ and $Q$ running in parallel; \emph{session restriction} $\newnt{s} P$ binds
session name $\ses{s}$ in process $P$; and $\0$ is the inactive process.

As in the $\pi$-calculus, but unlike other MPST calculi, $\MPSTb$ supports
\emph{output-guarded choice} $\sum_{i\in I}
\snd{c_i}{\role{\rho_i}}{\msgLabel{m}_i\langle \poly{V_i} \rangle} P_i$,
allowing a nondeterministic send along any $c_i$ to role $\role{\rho_i}$ with
label $\m_i$ and payload $\poly{V_i}$, with the process continuing as $P_i$.
\emph{Branching receive} $c[\role{\rho}] \&_{i \in I}\ \m_i (\poly{b_i}) \then P_i$ 
denotes a process waiting on channel $c$ for one of a set of messages from 
role $\role{\rho}$ with label $\m_i$, binding the received data to $\poly{b_i}$ 
before continuing according to $P_i$.
It is key to note that the object of a communication action is indicated via
$\role{\rho}$ (for both sending and receiving), which can either be a concrete
value, or a role variable.
The second key difference is the optional use of the \emph{bang} $!$ 
with a branching receive, marking it as \emph{replicated}---\ie\ it may
be used $0$ or more times, modelling \emph{infinitely available servers}.

\begin{figure}[t!]
    \headersig{Process reduction}{$P_1 \redP P_2$}
    {\small\begin{mathpar}
        \infer[\ \textnormal{if} $k \in I$]{R-C}{r-c}
        {}
        {\rcvI{\ses{s}[\role{p}]}{\role{q}}{\m_i(\poly{b_i})} P_i\ \|\ \snd{\ses{s}[\role{q}]}{\role{p}}{\m_k\langle \poly{d} \rangle} Q
        \ \redP\
        P_k \sub{\poly{d}}{\poly{b_k}}\ \|\ Q}

        \infer[\ \textnormal{if} $k \in I$]{R-!C$_1$}{r-bang1}
        { R = {!}\rcvI{\ses{s}[\role{p}]}{\role{q}}{\m_i(\poly{b_i})} Q_i }
        {\snd{\ses{s}[\role{q}]}{\role{p}}{\m_k\langle \poly{d} \rangle} P\ \|\ R
        \ \redP\
        P\ \|\ R\ \|\ Q_k \sub{\poly{d}}{\poly{b_k}}}

        \infer[\ \textnormal{if} $k \in I$]{R-!C$_2$}{r-bang2}
        { R = {!}\rcvI{\ses{s}[\role{p}]}{\role{\alpha}}{\m_i(\poly{b_i})} Q_i }
        {\snd{\ses{s}[\role{q}]}{\role{p}}{\m_k\langle \poly{d} \rangle} P\ \|\ R
        \ \redP\
        P\ \|\ R\ \|\ Q_k \sub{\poly{d}}{\poly{b_k}} \sub{\role{q}}{\role{\alpha}}}

        \infer[\textnormal{ for $j \in I$}]{R-+}{r-choice}
        {}
        {{\textstyle\sum_{i \in I}}\, \snd{\ses{s_i}[\role{q_i}]}{\role{p_i}}{\msgLabel{m}_i\langle \poly{d_i} \rangle} P_i
        \ \redP\ 
        \snd{\ses{s_j}[\role{q_j}]}{\role{p_j}}{\msgLabel{m}_j\langle \poly{d_j} \rangle} P_j}

        \infer{R-\recVar{X}}{r-proc-call}
        {}
        {\defin{X}{\poly{x}}{P\;} (\recVar{X}\langle \poly{\ses{s^\prime}[\role{r}]} \rangle \| Q)
        \ \redP\ 
        \defin{X}{\poly{x}}{P\;} (P\sub{\poly{\ses{s^\prime}[\role{r}]}}{\poly{x}} \| Q)}

        \infer{R-$\equiv$}{r-cong}
        {P \equiv P^\prime \\ P \redP Q \\ Q \equiv Q^\prime}
        {P^\prime \redP Q^\prime}

        \infer{R-$\rc$}{r-ctx}
        {P \redP P^\prime}
        {\rc[P] \redP \rc[P^\prime]}
    \end{mathpar}}
    \caption{Operational semantics for the extended multiparty session $\pi$-calculus.}
    \label{fig:os}
\end{figure}

%
\begin{definition}[Reduction context]\label{def:ctx}
    A \emph{reduction context} $\rc$ is given as: 
    $\rc ::= \rc\|P \mid \newnt{s}\,\rc \mid \definD{D\;}\rc \mid [\ ]$.
\end{definition}

A \emph{reduction context} allows us to evaluate processes under parallel composition and name restrictions.
With this, reduction rules on processes are given in \Cref{fig:os}. 
The rules make use
of a standard \emph{structural congruence} $\equiv$ (\Cref{app:cong}) that
allows us to treat parallel composition as commutative and associative, as well
as including the usual $\pi$-calculus scope extrusion rule.

Rule \Cref{rule:r-c} shows synchronous communication between two processes in
session $\ses{s}$. The first process, playing role $\p$, offers role $\q$ a
choice of message labels and associated process continuations. The second
process, playing role $\q$, sends a message with label $\msgLabel{m}_k$ and
transmits payloads $\poly{d}$. The first process reduces to the selected
continuation with the transmitted payloads substituted for the binders in the
selected branch, and the second process reduces to the continuation $Q$.

Rules \Cref{rule:r-bang1} and \Cref{rule:r-bang2} describe communication with a
\emph{replicated} process $R$.  Rule \Cref{rule:r-bang1} is similar to
\textsc{R-C} but the replicated process remains unchanged and the continuation
$Q_k$ is evaluated in parallel.  Rule \Cref{rule:r-bang2} handles the case where
the replicated process does not need to receive from a specific role, but
instead allows communication with an \emph{arbitrary} role: the rule binds the
sending role to $\role{\alpha}$ in the replicated continuation.
We refer to this as a \emph{universal receive}.

The \Cref{rule:r-choice} rule evaluates a branching output by nondeterministically
evaluating to one of the sending branches; rule \Cref{rule:r-proc-call} handles a recursive
call.
Finally, rules \Cref{rule:r-cong} and \Cref{rule:r-ctx} are administrative,
allowing reduction modulo structural congruence and under contexts respectively.

\begin{figure}[t]
    {\small\begin{align*}
        P_\role{c} &:= 
            \snd{\ses{s}[\role{c}]}{\role{s}}{\msgLabel{req} \langle 42 \rangle}
            \rcv{\ses{s}[\role{c}]}{\role s}{\msgLabel{wrk} (\role \omega)}
            \rcv{\ses{s}[\role{c}]}{\role \omega}{\msgLabel{ans} (z)} \0 && \\
        P_\role{s} &:= 
            {!}\rcv{\ses{s}[\role{s}]}{\role{\alpha}}{\msgLabel{req} (x)}
            {\sum_{i=1}^{2}} \snd{\ses{s}[\role{s}]}{\role{w_i}}{\msgLabel{fw}\langle x,\role \alpha \rangle}
            \snd{\ses{s}[\role{s}]}{\role{\alpha}}{\msgLabel{wrk}\langle \role{w_i} \rangle} \0 && \\
        P_\role{w_i} &:= 
            {!}\rcv{\ses{s}[\role{w_i}]}{\role{s}}{\msgLabel{fw} (y,\role \gamma)}
            \snd{\ses{s}[\role{w_i}]}{\role \gamma}{\msgLabel{ans} \langle \text{``life''} \rangle} \0 &&
    \end{align*}}
    \caption{Process definitions for load balancer example}
    \label{fig:balancer-procs}
\end{figure}

\begin{example}[Load Balancer: Process Reduction]\label{ex:lb-proc-red}
    We recall the load balancer example from \Cref{sec:intro}, but this time, we
    present processes for each role (\Cref{fig:balancer-procs}) to
    demonstrate how our operational semantics handles communication.
    Consider a single client $P_\role{c}$ in parallel with three server-side
    processes:%
    {\small
    \begin{flalign*}
            & P_\role{c} &\|\tab& P_\role{s}  &\|\tab&  P_\role{w_1}  &\|\tab& P_\role{w_2}   &&\\
        =\  & \snd{\ses{s}[\role{c}]}{\role{s}}{\msgLabel{req} \langle 42 \rangle} P_\role{c}^\prime  
            &\|\tab&
              {!}\rcv{\ses{s}[\role{s}]}{\role{\alpha}}{\msgLabel{req} (x)} P_\role{s}^\prime &\|\tab&  P_\role{w_1}  &\|\tab& P_\role{w_2} &&
    \end{flalign*}}%
    Using \Cref{rule:r-bang2}, the client and server reduce. 
    The reduction advances the client to its continuation $P_\role{c}^\prime$, and pulls out a copy of the server's continuation as a new process.
    It is key to note that $\role \alpha$ is acting as a binder in $P_\role{s}$, therefore, in the continuation we observe a role variable substitution:
    {\small
    \begin{flalign*}
        \redP\ & P_\role{c}^\prime &\|\ & P_\role{s} &\|\ & \left({\sum_{i=1}^{2}} \snd{\ses{s}[\role{s}]}{\role{w_i}}{\msgLabel{fw}\langle x,\role \alpha \rangle} P_\role{s_i}^{\prime\prime}\right)\  \sub{42}{x} \sub{\role c}{\role \alpha} &\|\ & P_\role{w_1}  &\|\ & P_\role{w_2}   &&\\
        =\ & P_\role{c}^\prime &\|\ & P_\role{s} &\|\ & {\sum_{i=1}^{2}} \snd{\ses{s}[\role{s}]}{\role{w_i}}{\msgLabel{fw}\langle 42,\role c \rangle} \left( P_\role{s_i}^{\prime\prime} \  \sub{42}{x} \sub{\role c}{\role \alpha}\right) &\|\ & P_\role{w_1}  &\|\ & P_\role{w_2}   &&
    \end{flalign*}
    }%
    The spawned server process will then non-deterministically choose a worker
    to send to via rule \Cref{rule:r-choice}; suppose $\role{w_1}$ is picked.
    {\small
    \begin{flalign*}
        \redP\ & P_\role{c}^\prime &\|\ & P_\role{s} &\|\ & \snd{\ses{s}[\role{s}]}{\role{w_1}}{\msgLabel{fw}\langle 42,\role c \rangle} P_\role{s_1}^{\prime\prime}  &\|\ & {!}\rcv{\ses{s}[\role{w_1}]}{\role{s}}{\msgLabel{fw} (y,\role \gamma)} P_\role{w_1}^\prime  &\|\ & P_\role{w_2}
    \end{flalign*}
    }%
    Communication is now possible between the spawned server process and worker $\role{w_1}$, using rule \Cref{rule:r-bang1}.
    As before, this communication advances the sender process and pulls out a copy of the worker's continuation:
    {\small
    \begin{flalign*}
        \redP\ & P_\role{c}^\prime &\|\tab& P_\role{s} &\|\tab& P_\role{s_1}^{\prime\prime}  &\|\tab&  P_\role{w_1} &\|\tab&  P_\role{w_1}^{\prime}\ \sub{42}{y}\sub{\role c}{\role \gamma}  &\|\tab& P_\role{w_2} &&
    \end{flalign*}}%
    The client can now learn which worker was chosen, 
    terminating the spawned server process, then the answer is exchanged 
    between the worker and client:
    {\small
    \begin{flalign*}
        \redP\ & 
        \rcv{\ses{s}[\role{c}]}{\role{w_1}}{\msgLabel{ans} (z)} \0
        &\|\ & P_\role{s} &\|\ & \0  &\|\ &  P_\role{w_1} &\|\ &  
        \snd{\ses{s}[\role{w_1}]}{\role c}{\msgLabel{ans} \langle \text{``life''} \rangle} \0 \!
        &\|\ & P_\role{w_2} && \\
        \redP\ & \0
        &\|\ & P_\role{s} &\|\ & \0  &\|\ &  P_\role{w_1} &\|\ & \0
        &\|\ & P_\role{w_2} && \\
        \equiv\ & P_\role{s} \ \|\  P_\role{w_1} \ \|\ P_\role{w_2}
    \end{flalign*}}
\end{example}

%% file: types.tex
\subsection{Types}\label{sec:types}

Figure~\ref{fig:type-syntax} shows the syntax of \MPSTb types. 

\paragraph{Syntax of types.}
Session types $\St$ type communication endpoints. They consist of
branching types $\type{S^\&}$, 
\emph{replicated} branching types $\type{!(S^\&)}$,
selection types $\type{S^\o}$,
recursive types $\type{\mu \recVar{t}.  S}$ and variables
$\type{\recVar{t}}$, and the
$\tEnd$ type.

Branching session types have the form
$\type{ \role{\rho} \&_{i \in I}\m_i (\poly{T_i}) . S_i}$ indicating that role
$\role{\rho}$ offers a choice of message labels $\type{\m_i}$ with payload types
$\type{\poly{T_i}}$ and continuations $\type{S_i}$.
Similarly, selection session types 
$\type{\o_{i \in I}\role{\rho_i}\,\m_i (\poly{T_i}) . S_i}$ indicate 
an internal choice towards one of a set of roles $\role{\rho_i}$, with a message label, given payload types and
continues as the corresponding continuation type.
A \emph{replicated} branching type $\type{!\role{\rho} \&_{i \in I}\m_i
(\poly{T_i}) . S_i}$ types a replicated channel. 
As with processes, role $\role{\rho}$ can either
be a concrete role, or a role variable in \emph{binding} position.

Our type system supports \emph{singleton} types for roles: role
$\role{\rho}$ has singleton type~$\role{\rho}$, used to pattern-match
specific roles in payloads.
\emph{Static} types $\type{T}$ are used for protocol
specification, whereas \emph{runtime} types $\type{U}$ are used by the
type semantics to include a notion of \emph{parallel composition} 
at type level:
originally introduced by Le~Brun and Dardha~\cite{DBLP:conf/forte/BrunD24},
a type $\type{U_1 \| U_2}$ 
allows a channel to be associated with multiple
``active'' session types.

We assume that branching and selection types include a non-empty set of messages
with pairwise distinct message names $\m_i$ (per role for $\type{\o}$). We further take an
\emph{equi-recursive} view of types, identifying a recursive type with its
unfolding (i.e., $\type{\mu \recVar{t} . S} = \type{S} \{ \type{\mu \recVar{t} .
S} / \type{\recVar{t}} \}$) and require that recursion variables are guarded.

\begin{figure}[t]
    {\small\begin{align*}
        \textit{Session types}& & \type{S} &::= \type{S^\&} \mid \type{{!} (S^\&)} \mid \type{S^\o} \mid \type{\mu \recVar{t} . S} \mid \type{\recVar{t}}  \mid \tEnd \\
        \textit{Branching/Selection Types}& &\type{S^\&} &::= \type{ \role{\rho} \&_{i \in I}\m_i (\poly{T_i}) . S_i} \tab\tab 
        \type{S^\o} ::= \type{\o_{i \in I}\role{\rho_i}\,\m_i (\poly{T_i}) . S_i} \\
        \textit{Role singletons}& & \role{\rho} &::= \role{q} \mid
        \role{\alpha}\\
        \textit{Static types}& & \type{T} &::= \type{S} \mid \role{\rho} \\ 
        \textit{Runtime types}& & \type{U} &::= \type{S} \mid \type{(U_1 \| U_2)}
\end{align*}}
\caption{Syntax of Types}
\label{fig:type-syntax}
\end{figure}

\begin{definition}[Subtyping]\label{def:subtyping}
    The \textbf{subtyping relation} $\subT$ is co-inductively defined on types by the following inference rules:
    {\normalfont\small\begin{mathpar}\mprset{fraction={===}}
        \inferrule
        {\type{(\poly{T_i} \,\subT\, \poly{T_i^\prime})_{i \in I}} \\ \type{(S_i\,\subT\, S_i^\prime)_{i \in I}}}
        {\type{\role{\rho} \&_{i \in I}\m_i (\poly{T_i}) . S_i \,\subT\, \role{\rho} \&_{i \in I \cup J}\m_i (\poly{T_i^\prime}) . S_i^\prime}}
        
        \inferrule
        {\type{(\poly{T_i} \supT \poly{T_i^\prime})_{i \in I}} \\ \type{(S_i\,\subT\, S_i^\prime)_{i \in I}}}
        {\type{\o_{i \in I\cup J}\role{\rho_i}\,\m_i (\poly{T_i}) . S_i \,\subT\, \o_{i \in I}\role{\rho_i}\,\m_i (\poly{T_i^\prime}) . S_i^\prime}}

        \inferrule
        {\type{S^\&_1 \,\subT\, S^\&_2}}
        {\type{{!}S^\&_1 \,\subT\, {!}S^\&_2}}
        
        \inferrule
        {\type{U \,\subT\, U^\prime} \\ \type{S \,\subT\, S^\prime}}
        {\type{U \| S \,\subT\, U^\prime \| S^\prime}}

        \inferrule
        {\type{S}\sub{\type{\mu\recVar{t}.S}}{\type{\recVar{t}}}\; \type{\,\subT\, S^\prime}}
        {\type{\mu\recVar{t}.S \,\subT\, S^\prime}}

        \inferrule
        {\type{S\,\subT\, S^\prime}\sub{\type{\mu\recVar{t}.S^\prime}}{\type{\recVar{t}}}}
        {\type{S \,\subT\, \mu\recVar{t}.S^\prime}}

        \inferrule
        { }
        {\type{T \,\subT\, T}}
    \end{mathpar}}
\end{definition}

Our definition of subtyping (\Cref{def:subtyping}) is mostly standard.
We adopt the convention of smaller types being ones with less external choice and more internal choice (\emph{à la} Gay and Hole~\cite{DBLP:journals/acta/GayH05}).
Subtyping of replication is based on regular branching; and parallel types are related iff their session types are subtypes.
Subtypes are related up-to their recursive unfolding, and subtyping is \emph{reflexive}.

\begin{definition}[Type Congruence]\label{def:type-cong}
    \emph{Type congruence} allows us to treat parallel runtime types as commutative and associative with identity element \textnormal{$\tEnd$}.
    {\normalfont\small\begin{mathpar}
    \inferrule{}{\type{U_1\|U_2 \equiv U_2\|U_1}}
    
    \inferrule{}{\type{(U_1\|U_2)\|U_3 \equiv U_1\|(U_2\|U_3)}} 

    \inferrule{}{\type{U\|\tEnd \equiv U}}
    \end{mathpar}}
\end{definition}

\begin{figure}[t]
    \header{Typing contexts}
    \[\begin{array}{l}
        \Theta ::= \emptyset \!\mid\! \Theta, X \hasT \type{\poly{S}}
        \tab\tab
        \Gamma ::= \emptyset \mid  
                   \Gamma, c \hasT \type{U}  \mid 
                   \Gamma, \role{\alpha} \hasT \role{\alpha} 
    \end{array}\]

    \headersig{Context splitting}{$\Gamma = \Gammab{1} \cSplit \Gammab{2}$}
    \begin{mathpar}
        \infer{}{split-0}
        { }
        {\emptyset = \emptyset \cSplit \emptyset}

        \infer{}{split-l1}
        {\Gamma = \Gammab{1} \cSplit \Gammab{2}}
        {\Gamma,c \hasT \type{U} = \Gammab{1},c \hasT \type{U} \cSplit \Gammab{2}}

        \infer{}{split-r1}
        {\Gamma = \Gammab{1} \cSplit \Gammab{2}}
        {\Gamma,c \hasT \type{U} = \Gammab{1} \cSplit \Gammab{2},c \hasT \type{U}}

        \infer{}{split-role}
        {\Gamma = \Gammab{1} \cSplit \Gammab{2}}
        {\Gamma,\role{\alpha} \hasT \role{\alpha} = \Gammab{1},\role{\alpha} \hasT \role{\alpha} \cSplit \Gammab{2},\role{\alpha} \hasT \role{\alpha}}

        \infer{}{split-par}
        {\Gamma = \Gammab{1} \cSplit \Gammab{2}}
        {\Gamma,c \hasT \type{U_1\|U_2} = \Gammab{1},c\hasT \type{U_1} \cSplit \Gammab{2},c \hasT \type{U_2}}
    \end{mathpar}

    \headersig{Context addition}{$\Gammab{1} + \Gammab{2} = \Gamma$}
    \begin{mathpar}
        \infer{}{+-empty}
        { }
        {\Gamma + \emptyset  = \Gamma}
        \ 
        \infer{}{+-Vl}
        {\Gammab{1} + \Gammab{2} = \Gamma \\ a \not\in \dom(\Gammab{2})}
        {\Gammab{1}, a \hasT \type T + \Gammab{2}  = \Gamma, a \hasT \type T}
        \ 
        \infer{}{+-Vr}
        {\Gammab{1} + \Gammab{2} = \Gamma \\ a \not\in \dom(\Gammab{1})}
        {\Gammab{1} + \Gammab{2}, a \hasT \type T  = \Gamma, a \hasT \type T}

        \infer{}{+-role}
        {\Gammab{1} + \Gammab{2} = \Gamma }
        {\Gammab{1},\role{\alpha} \hasT \role{\alpha} + \Gammab{2},\role{\alpha} \hasT \role{\alpha}  = \Gamma,\role{\alpha} \hasT \role{\alpha}}

        \infer{}{+-par}
        {\Gammab{1} + \Gammab{2} = \Gamma }
        {\Gammab{1},c:\type{U_1} + \Gammab{2},c:\type{U_2} = \Gamma,c:\type{U_1\|U_2}}
    \end{mathpar}

    \headersig{Context role insertion}{$\type{\Gamma} \insnew{\role{\rho}}$}
    \begin{mathpar}
        \type{\Gamma} \insnew{\role{q}} = \type{\Gamma}

        \type{\Gamma} \insnew{\role{\alpha}} = \type{\Gamma} + \role{\alpha} : \role{\alpha}
    \end{mathpar}

    \caption{Typing contexts and context operations.}
    \label{fig:context-operations}
\end{figure}

Figure~\ref{fig:context-operations} shows the definition of typing contexts
and their operations.
Context $\Theta$ is used to type recursive process definitions, mapping process
variables $X$ to tuples of parameter types. Context $\Gamma$ maps channels to
types, and role variable singletons. Context
composition $\Gamma,\Gammap$ is defined iff $\Gamma$ and $\Gammap$
have disjoint domains. We lift subtyping and type congruence to contexts in the
usual way.

As inspired by (e.g.)~\cite{DBLP:journals/iandc/Vasconcelos12}, linearity is
enforced through the use of a \emph{context split} operation $\Gamma = \type{\Gamma_1\cSplit \Gamma_2}$
that splits a context $\Gamma$ into two environments
$\type{\Gamma_1}$ and $\type{\Gamma_2}$. These environments may share variables with
unrestricted types and role variables.  Additionally, a channel $c$ with
runtime type $\type{U_1 \| U_2}$ may be split such that $\type{\Gamma_1}$ contains $c : \type{U_1}$
and $\type{\Gamma_2}$ contains $c : \type{U_2}$; this allows us to type endpoints used by
different replicated processes.
The inverse operation is \emph{context addition} $\type{\Gamma_1} + \type{\Gamma_2} = \Gamma$
that combines $\type{\Gamma_1}$ and $\type{\Gamma_2}$ into an environment $\Gamma$; again
roles may be shared across environments. In the case that we have $\type{\Gamma_1}, c :
\type{U_1} + \type{\Gamma_2}, c : \type{U_2}$ (i.e., a linear variable used in two different
contexts), context addition results in $c$ having the combined runtime type
$\type{U_1 \| U_2}$.
The \emph{role insertion} operation $\type{\Gamma} \insnew{\role{\rho}}$ is used when
typing replicated receives and extends a context with a mapping $\role{\alpha} :
\role{\alpha}$ in the case that $\role{\rho}$ is a role variable, and leaves
$\type{\Gamma}$ unchanged otherwise.


Before presenting the typing rules, we introduce the notion of
a \emph{session protocol}, mapping role names to session types. 
%
The $\extract_\ses{s}(\Psi)$ function associates a session protocol $\Psi$ to a 
concrete session $\ses{s}$ to form a typing context.

\begin{definition}[Session Protocol]\label{def:protocol}
    A session protocol $\Psi$ is attached to a session through the restriction operation in the form of $\new{s}{\Psi} P$, dictating the protocol for $\ses{s}$ in $P$.
    A protocol $\Psi$ is a partial mapping from role to session type, given as:
    {\normalfont\[\Psi ::= \emptyset \mid \Psi,\q : \type{S}\]}
    A protocol $\Psi$ is well-formed iff it does not contain parallel types.
    We obtain a typing context by associating roles in a protocol 
    with a session name.
    Formally:
    {\normalfont\small\[
        \extract_\ses{s}(\q : \St, \Psi) := \ses{s}[\q] \hasT \St, \extract_\ses{s}(\Psi)
        \tab\tab
        \extract_\ses{s}(\emptyset) := \emptyset
    \]
    }
\end{definition}

Next, we introduce predicate $\pEnd$ on type contexts.
This ensures that an environment is \emph{end-typed}: that is, its  session types are
congruent to type $\tEnd$.

\begin{definition}[End-typed environment]\label{def:end}
    A context is \emph{end-typed}, written \textnormal{$\pEnd(\Gamma)$}, iff it
    only maps channels to session type \textnormal{$\tEnd$}.
    %
    %
    {\normalfont\begin{mathpar}
        \inferrule
        { }
        {\pEnd(\emptyset)}

        \inferrule
        {\pEnd(\Gamma)}
        {\pEnd(c \hasT \tEnd, \Gamma)}

        \inferrule
        {\pEnd(\Gamma)}
        {\pEnd(\role{\alpha} \hasT \role{\alpha}, \Gamma)}

        \inferrule
        {\Gamma \type{\equiv} \Gammap \ \ \pEnd(\Gamma)}
        {\pEnd(\Gammap)}
    \end{mathpar}}
    %
        
\end{definition}


\begin{figure}[t]
        \setlength{\lineskip}{0pt} 
        \mprset{sep=1em}
        \headertwo
            {Typing Rules for Values and Recursive Definitions}
            {\framebox{$\Gamma \vdash V \hasT \type{T}\vphantom{\poly{S}}$}~
             \framebox{$\Theta \vdash X \hasT \type{\poly{S}}$}}
        {\normalfont\small\begin{mathpar}
            \infer{T-Wkn}{t-weak}
            {\Gammab{1} + \Gammab{2} = \Gamma \\\\ \Gammab{1} \vdash V \hasT \type{T} \\ \pEnd(\Gammab{2}) }
            {\Gamma \vdash V \hasT \type{T}}
            \hspace{-1.5em}
            \and
            \infer{T-$\role{q}$}{t-roleVal}
            { }
            {\emptyset \vdash \role{q} \hasTT \q}
            \hspace{-1.5em}
            \and
            \infer{T-Sub}{t-sub}
            {\type{S \subT S^\prime}}
            {c \hasTT \type{S} \vdash c \hasTT \type{S^\prime}}
            \hspace{-1.5em}
            \and
            \infer{T-$\role{\alpha}$}{t-roleVar}
            { }
            {\role \alpha \hasTT \role \alpha \vdash \role \alpha \hasTT \role \alpha}
            \hspace{-1.5em}
            \and
            \infer{T-$\recVar{X}$}{t-x}
            {\Theta(\recVar{X}) {=} \type{(S_i)_{i \in 1..n}}}
            {\Theta \vdash \recVar{X} \hasTT \type{(S_i)_{i \in 1..n}}}
        \end{mathpar}}
    
    \headersig{Typing Rules for Processes}{$\Theta;\Gamma \vdash P$}
    {\normalfont\small\begin{mathpar}
        \infer{T-$\0$}{t-0}
        {\pEnd(\Gamma)}
        {\Theta \cons \Gamma \vdash 0}
        \ 
        \infer{T-{$|$}}{t-par}
        {\Theta \cons \type{\Gamma_1} \vdash P_1 \\ \Theta \cons \type{\Gamma_2} \vdash P_2}
        {\Theta \cons \type{\Gamma_1} \cSplit \type{\Gamma_2} \vdash P_1 \| P_2}
        \ 
        \infer{T-{$+$}}{t-choice}
        {(\Theta \cons \Gamma \vdash \snd{c_i}{\role{\role{\rho_i}}}{\msgLabel{m}_i\langle \poly{d_i} \rangle} P_i)_{i \in I}}
        {\Theta \cons \Gamma \vdash {\textstyle\sum_{i \in I}}\, \snd{c_i}{\role{\role{\rho_i}}}{\msgLabel{m}_i\langle \poly{d_i} \rangle} P_i}

        \infer{T-\type{$!$}}{t-bang}
        {
        \type{\Gamma_!} \vdash c \hasT \type{!\role \rho \&_{i \in I}\m_i (\poly{T_i}) . S_i^\prime} \\
        \pEnd(\Gamma) \\\\
        (\Theta \cons \Gamma + c \hasT \type{S_i^\prime} + \poly{b_i} :
        \type{\poly{T_i}} \insnew{\role \rho} \vdash P_i)_{i\in I}
        }
        {\Theta \cons \Gamma \cSplit \type{\Gamma_!} \vdash {!}c[\role \rho]\ \oldAnd_{i\in I}\, \m_i (\poly{b_i}) \then P_i}
        \infer{T-\type{$\&$}}{t-rcv}
        {
        \type{\Gamma_\&} \vdash c \hasT \type{\role \rho \&_{i \in I}\m_i (\poly{T_i}) . S_i^\prime} \\
        \type{\Gamma_r} \vdash \role \rho \hasT \role \rho \\\\
        (\Theta \cons \Gamma + c \hasT \type{S_i^\prime} + \poly{b_i} \hasT \type{\poly{T_i}} \vdash P_i)_{i \in I}
        }
        {\Theta \cons \Gamma \cSplit\type{\Gamma_\&} \cSplit \type{\Gamma_r} \vdash c[\role \rho]\ \oldAnd_{i\in I}\, \m_i (\poly{b_i}) \then P_i}
        
        \infer{T-{$\nu$}}{t-new}
        {\varphi(\extract_\ses{s}(\Psi)) \\\\
          \textnormal{ $\varphi$ is a \emph{safety} property} \\\\ \ses{s} \not\in \Gamma \\ \Theta \cons \Gamma + \extract_\ses{s}(\Psi) \vdash P}
        {\Theta \cons \Gamma \vdash \new{s}{\type{\Psi}} P} 
        \infer{T-\type{$\o$}}{t-send}
        {
        \type{\Gamma_\o} \vdash c \hasT \type{\role{\rho} {\o} \m (\poly{T}) . S^\prime} \\
        \type{\Gamma_r} \vdash \role{\rho} \hasT \role \rho \\\\
        (\type{\Gamma_i} \vdash V_i \hasT \type{T_i})_{i \in 1..n} \\
        \Theta \cons \Gamma + c \hasT \type{S^\prime}  \vdash P
        }
        {\Theta \cons \Gamma \cSplit \type{\Gamma_\o} \cSplit \type{\Gamma_r} (\!\cSplit \type{\Gamma_i})_{i\in 1..n}
        \vdash 
        c[\role{\rho}] \o \m \langle (V_i)_{i \in 1..n} \rangle \,.\, P}        

        \infer[]{T-Def}{t-def}
        {
        \Theta, \recVar{X}\hasT \type{\poly{S}} \cons \Gamma, \poly{x \hasT \type{S}} \vdash P \\
        \pEnd(\Gamma) \\\\
        \Theta, \recVar{X}\hasT \type{\poly{S}} \cons \Gammap \vdash Q
        }
        {\Theta \cons \Gamma \cSplit \Gammap \vdash \defin{X}{\poly{x \hasT \type{S}}}{P\;} Q}
        \quad
        \infer[]{T-Call}{t-call}
        {\Theta \vdash \recVar{X} \hasT (\type{S_i})_{i \in 1..n} \\ \pEnd(\type{\Gamma_0}) \\\\
        \forall i \in 1..n : \type{\Gamma_i} \vdash c_i \hasT \type{S_i} }
        {\Theta \cons \type{\Gamma_0}( \cSplit \type{\Gamma_i})_{i \in 1..n} \vdash \recVar{X}\langle (c_i)_{i \in 1..n} \rangle}
    \end{mathpar}}

    \caption{Typing rules.}
    \label{fig:typing-rules}
\end{figure}


Figure~\ref{fig:typing-rules} shows the typing rules for \MPSTb. There are three
judgements: the value typing judgement $\type{\Gamma} \vdash V : \type{T}$
assigns type $\type{T}$ to value $V$ under context $\type{\Gamma}$ and consists
of four rules. Rule \Cref{rule:t-weak} allows weakening: if value $V$ has type
$\type{T}$ under some environment $\type{\Gamma_1}$, and we have an end-typed
environment $\type{\Gamma_2}$ such that $\type{\Gamma_1} + \type{\Gamma_2} = \Gamma$, then the
rule allows us to conclude that $V$ has type $\type{T}$ under $\type{\Gamma}$. 
Rule \Cref{rule:t-roleVal} types concrete roles as singleton types under the empty
environment, whereas \Cref{rule:t-roleVar} types a role variable provided
that it is contained within the type environment. 
Finally, rule \Cref{rule:t-sub}
types a linear variable, allowing for subtyping. Judgement $\type{\Theta} \vdash X :
\type{\poly{S}}$ simply looks up process variable $X$ in process environment
$\type{\Theta}$, returning its parameter types; this is achieved by rule \Cref{rule:t-x}.

The final judgement $\type{\Theta} \cons \type{\Gamma} \vdash P$ states that
process $P$ is typable under recursion environment $\type{\Theta}$ and type
environment $\type{\Gamma}$.
Rule \Cref{rule:t-0} types the inactive process under an end-typed environment.
Rule \Cref{rule:t-par} types parallel processes under a
split environment. Rule \Cref{rule:t-choice} types an output-directed choice; since this
operation uses branching control flow, each choice must be typable under the
same environment.

Rules \Cref{rule:t-bang} and \Cref{rule:t-rcv} type replicated and non-replicated receives
respectively. In both cases the rules check that the channel has a receiving
session type. Both rules check that each branch is typable by extending a common
typing context with the variables bound by the receives, along with the
continuation type for the session channel.
In the case of a replicated receive there are two differences: first, the
context used to type each branch must be end-typed (in order to avoid
duplicating linear resources). Second, since the role $\role{\rho}$ may be a
\emph{binding} occurrence of a role variable $\role{\alpha}$, the context used
to type each branch must be extended with the role variable (if applicable)
using the insertion operator.

Rule \Cref{rule:t-new} types a session name restriction $(\nu \ses{s} : \Psi) P$. 
As is standard in generalised MPST~\cite{DBLP:journals/pacmpl/ScalasY19}, the session protocol must satisfy a \emph{safety property} $\varphi$.
We discuss specifics of this property in \Cref{sec:type-sem}, and how it is used in \Cref{sec:meta}.
Informally, the property 
ensures that all process communication is ``compatible'', and is the weakest
property required to prove subject reduction.
%
We then prove our metatheory parametric on the \emph{largest} safety property, 
allowing it to be customised to verify more precise properties (e.g., to ensure deadlock-freedom or
termination), based on the requirements of a specific protocol. 
%

Rule \Cref{rule:t-send} types an output if the sending channel 
can be mapped to a selection type with payload types that match the values being sent.
The rule ensures that role $\role \rho$ is either a 
value, or it exists in the type context---\ie\ messages cannot be sent to unbound role variables.
The selection type continuation should be used, along with the common context, to type the continuation process.

Lastly, rules \Cref{rule:t-def} and \Cref{rule:t-call} type recursive processes.
The former populates the recursion environment whilst ensuring that process declarations are well typed under the variables they bind.
The latter checks that types of values used in a process call match what was specified in the declaration.


\begin{example}[Load Balancer: Type Checking]\label{ex:lb-type-check}
    By introducing a name restriction for session $\ses{s}$ that includes the
    types from~\Cref{ex:lb-intro}, we can type the processes from 
    \Cref{ex:lb-proc-red} under empty typing contexts:
    {\normalfont\small\[\emptyset;\emptyset \vdash \new{\ses{s}}{\{\role{s}\hasT\type{S_\role{s}}, \role{w_1}\hasT\type{S_\role{w_1}}, \role{w_2}\hasT\type{S_\role{w_2}}, \role{c}\hasT\type{S_\role{c}}\}} P_\role{s} \| P_\role{w_1} \| P_\role{w_2} \| P_\role{c}\]}
\end{example}


\subsection{Type Semantics}\label{sec:type-sem}

\begin{figure}[t]
    {\normalfont\small
        \[\text{Actions} \quad \act ::= \aOut{s}{\q}{\role r}{\m}{\poly{T}} \mid \aIn{s}{\q}{\role \rho}{\m}{\tilde{T}} \mid \aCom{s}{\q}{\role r}{\m}\]

        \headersig{Context Reduction}{$\Gamma \redC \Gamma'$}
        \begin{mathpar}
        \infer{$\Gamma$-\type{$\&$}}{g-rcv}
        {\St = \type{\q \&_{i \in I}\, \m_i (\poly{T_i}) . S_i^\prime} \\ k \in I}
        {\ses{s}[\p] \hasT \St \redS{\aIn{s}{\p}{\q}{\m_k}{\poly{T_k}}}  \ses{s}[\p] \hasT \type{S_k^\prime}}
        \,
        \infer{$\Gamma$-\type{$\o$}}{g-send}
        {\St = \type{{\o_{i \in I}} \role{q_i} \m_i (\poly{T_i}) . S_i^\prime} \\ k \in I}
        {\ses{s}[\p] \hasT \St \redS{\aOut{s}{\p}{\role{q_k}}{\m_k}{\poly{T_k}}} \ses{s}[\p] \hasT \type{S_k^\prime}}
        \,
        \infer{$\Gamma$-$\role \rho$}{g-role}
        {\Gamma \redS{\act} \type{\Gamma^\prime}}
        {\Gamma,\role \rho \hasT \role \rho \redS{\act} \type{\Gamma^\prime},\role \rho \hasT \role \rho}

        \infer{$\Gamma$-\type{$!$}}{g-bang}
        {\type{R} = \type{!\role \rho \&_{i \in I}\, \m_i (\poly{T_i}) . S_i} \\ k \in I}
        {\ses{s}[\q] \hasT \type{R} \redS{\aIn{s}{\q}{\role \rho}{\m_k}{\poly{T_k}}}  \ses{s}[\q] \hasT \type{R \| S_k}}

        \infer{$\Gamma$-Cong}{g-cong}
        {\Gamma \redS{\act} \type{\Gamma^\prime}}
        {\Gamma \cSplit \type{\Gamma^{\prime\prime}} \redS{\act} \type{\Gamma^\prime} {+} \type{\Gamma^{\prime\prime}}}

        \infer{$\Gamma$-\type{$\mu$}}{g-rec}
        {\Gamma \cSplit c \hasT \type{S} \sub{\type{\mu\recVar{t}.S}}{\type{\recVar{t}}}\, \redS{\act}\, \type{\Gamma^\prime}}
        {\Gamma \cSplit c \hasT \type{\mu\recVar{t}.S}\, \redS{\act}\, \type{\Gamma^\prime}}

        \infer{$\Gamma$-Com$_1$}{g-com1}
        {\Gamma = \type{\Gamma_{1}} \cSplit \type{\Gamma_{2}} \\ \type{\Gamma_1} \redS{\aOut{s}{\p}{\q}{\m}{\poly{S}, \poly{\role{r}} }} \type{\Gamma_1^\prime} \\ \type{\Gamma_2} \redS{\aIn{s}{\q}{\p}{\m}{\poly{S}^\prime, \poly{\role{\alpha}} }} \type{\Gamma_2^\prime} \\ \type{\poly{S} \subT \poly{S}^\prime}}
        {\Gamma \redS{\aCom{s}{\p}{\q}{\m}} \type{\Gamma_{1}^\prime} + \type{\Gamma_{2}^\prime} \sub{\type{\poly{\role{r}}}}{\type{\poly{\role{\alpha}}}} }

        \infer{$\Gamma$-Com$_2$}{g-com2}
        {\Gamma = \type{\Gamma_{1}} \cSplit \type{\Gamma_{2}} \\ \type{\Gamma_1} \redS{\aOut{s}{\p}{\q}{\m}{\poly{S}, \poly{\role{r}} }} \type{\Gamma_1^\prime} \\ \type{\Gamma_2} \redS{\aIn{s}{\q}{\role{\alpha_0}}{\m}{\poly{S}^\prime, \poly{\role{\alpha}} }} \type{\Gamma_2^\prime} \\ \type{\poly{S} \subT \poly{S}^\prime}}
        {\Gamma \redS{\aCom{s}{\p}{\q}{\m}} \type{\Gamma_{1}^\prime} + \type{\Gamma_{2}^\prime}\sub{\p}{\role{\alpha_0}} \sub{\type{\poly{\role{r}}}}{\type{\poly{\role{\alpha}}}} }
    \end{mathpar}}
    \caption{Type semantics.}
    \label{fig:type-sem}
\end{figure}

To reason about the interactions between session types, following Scalas \&
Yoshida~\cite{DBLP:journals/pacmpl/ScalasY19}, we endow typing contexts $\Gamma$
with LTS semantics as shown in \Cref{fig:type-sem}.
Each action $\act$ denotes an \emph{output}, \emph{input}, and
\emph{synchronising communication} respectively.
\emph{Context reduction} \textnormal{$\Gamma \redC \Gammap$} is defined iff \textnormal{$\Gamma \redS{\aCom{s}{\p}{\q}{\m}} \Gammap$} for some \textnormal{$\ses{s},\p,\q,\m$},
and we write $\redC^*$ for the transitive and reflexive closure of $\redC$.
We write $\Gamma \redC$ iff $\Gamma \redC \Gammap$ for some $\Gammap$.

Transitions \Cref{rule:g-rcv} and \Cref{rule:g-send} are standard: 
a context can fire an input label (resp. output label) matching any 
of the labels in the top-level branch type (resp. selection type).
This transitions the entry to the continuation of the chosen type.

%
%
Transition \Cref{rule:g-bang} models the receipt of a message by a replicated input.
The two main differences to the linear receive are:
\emph{(i)}~the role $\role \rho$ used in the transition label is allowed to be a role variable name; and
\emph{(ii)}~firing an input does not advance the type, but instead pulls out a copy of the continuation and places it in parallel.
Role $\role \rho$ is considered bound in $\type{R}$ and its continuations, but is \emph{free} in pulled out copies of the continuations composed in parallel (\type{$S_k$}).

Transition rules \Cref{rule:g-role}, \Cref{rule:g-cong}, \Cref{rule:g-rec} allow contexts to reduce under a larger context, or when types are guarded by recursive binders.
Concretely, \Cref{rule:g-role} allows transitions to ignore role singletons; \Cref{rule:g-cong} allows a context to perform a transition when split from a larger context (the transition result must be added back in); and \Cref{rule:g-rec} allows recursive binders to be treated equal to their unfolding.

Transitions \Cref{rule:g-com1} and \Cref{rule:g-com2} model type-level
communication; for simplicity and without loss of generality we assume a
convention wherein session-typed payloads precede role-typed payloads.
Both rules state that if a context can be split such that one part fires an output label, and the other fires an input with matching roles, message label and payloads, then the entire context can transition via a \emph{communication} action.
We note that payloads will consist of either session types or role singletons. 
For the former, sender payloads must be subtypes of the receiver payloads; for the latter, \emph{role substitution} occurs after communication.
\Cref{rule:g-com2} caters for universal receives, where the input label consists of a role variable in binding position---this is reflected in the role substitution.
We say that a context \emph{reduces} iff it can transition via communication actions.

\begin{example}[Load Balancer: Context Reduction]
    We now use the protocol from \Cref{ex:lb-type-check} to demonstrate context reduction in action.
    Initially, the only communication action possible is between the client and server, via \Cref{rule:g-com2}.
    {\small\[\begin{array}{r l}
        & \ses{s}[\role{c}] \hasT \type{S_\role{c}},\ \ses{s}[\role{s}] \hasT \type{S_\role{s}},\ \ses{s}[\role{w_1}] \hasT \type{S_\role{w_1}},\ \ses{s}[\role{w_2}] \hasT \type{S_\role{w_2}} \\
        =\ & \\ 
        & (\ses{s}[\role{c}] \hasT \type{S_\role{c}} \cSplit \ses{s}[\role{s}] \hasT \type{S_\role{s}}) \cSplit (\ses{s}[\role{w_1}] \hasT \type{S_\role{w_1}},\ \ses{s}[\role{w_2}] \hasT \type{S_\role{w_2}}) \\
        \redC\ & \\ 
        &   \ses{s}[\role{c}] \hasT \type{\role s {\&} \msgLabel{wrk}(\role \omega) .\, \role \omega \& \msgLabel{ans}(\texttt{str})}
        + \left(\ses{s}[\role{s}] \hasT \type{S_\role{s} \!\!\!\!\!\!\quad\biggm\lvert {\o} \left\{\begin{array}{l}
                \role{w_1} \msgLabel{fw}(\texttt{int}, \role \alpha) .\, \role \alpha {\o} \msgLabel{wrk}(\role{w_1}) \\
                \role{w_2} \msgLabel{fw}(\texttt{int}, \role \alpha) .\, \role \alpha {\o} \msgLabel{wrk}(\role{w_2})
            \end{array}\right.}\right)\!\sub{\role c}{\role \alpha} \\
        & +\ \ses{s}[\role{w_1}] \hasT \type{S_\role{w_1}},\ \ses{s}[\role{w_2}] \hasT \type{S_\role{w_2}} \\
        =\ & \\
        &   \ses{s}[\role{c}] \hasT \type{\role s {\&} \msgLabel{wrk}(\role \omega) .\, \role \omega \& \msgLabel{ans}(\texttt{str})}, \\
        &   \ses{s}[\role{s}] \hasT \type{S_\role{s} \!\!\!\!\quad\biggm\lvert {\o} \left\{\begin{array}{l}
                \role{w_1} \msgLabel{fw}(\texttt{int}, \role c) .\, \role c {\o} \msgLabel{wrk}(\role{w_1}) \\
                \role{w_2} \msgLabel{fw}(\texttt{int}, \role c) .\, \role c {\o} \msgLabel{wrk}(\role{w_2})
            \end{array}\right.},\ 
            \ses{s}[\role{w_1}] \hasT \type{S_\role{w_1}},\ \ses{s}[\role{w_2}] \hasT \type{S_\role{w_2}}
    \end{array}\]}
    
    \noindent
    It is key to note the role substitution for the type of $\ses{s}[\role s]$
    above; specifically, how the substitution affects the parallel type
    extracted through communication, but does not affect the replicated type $\type{S_\role{s}}$.
    From here, there are multiple reduction paths, but let us observe the
    reduction path in which $\role{s}$ communicates with $\role{w_1}$.
    {\normalfont\small\[\begin{array}{r l}
        =\ & \\
        &   \ses{s}[\role{c}] \hasT \type{\role s {\&} \msgLabel{wrk}(\role \omega) .\, \role \omega \& \msgLabel{ans}(\texttt{str})},\ \ses{s}[\role{w_2}] \hasT \type{S_\role{w_2}},\ \ses{s}[\role{s}] \hasT \type{S_\role{s}} \\
        &   \cSplit \left(\ses{s}[\role{s}] \hasT \type{{\o} \left\{\begin{array}{l}
                \role{w_1} \msgLabel{fw}(\texttt{int}, \role c) .\, \role c {\o} \msgLabel{wrk}(\role{w_1}) \\
                \role{w_2} \msgLabel{fw}(\texttt{int}, \role c) .\, \role c {\o} \msgLabel{wrk}(\role{w_2})
            \end{array}\right.} \cSplit
            \ses{s}[\role{w_1}] \hasT \type{S_\role{w_1}} \right) \\
        \redC\ & \\ 
        &   \ses{s}[\role{c}] \hasT \type{\role s {\&} \msgLabel{wrk}(\role \omega) .\, \role \omega \& \msgLabel{ans}(\texttt{str})},\ \ses{s}[\role{w_2}] \hasT \type{S_\role{w_2}},\ \ses{s}[\role{s}] \hasT \type{S_\role{s}} \\
        & + \left(\ses{s}[\role{s}] \hasT \type{\role c {\o} \msgLabel{wrk}(\role{w_1})} + (\ses{s}[\role{w_1}] \hasT \type{S_\role{w_1} \,\|\, \role \gamma {\o} \msgLabel{ans}(\texttt{str})}) \sub{\role c}{\role \gamma} \right)\\
        =\ & \\
        &   \ses{s}[\role{c}] \hasT \type{\role s {\&} \msgLabel{wrk}(\role \omega) .\, \role \omega \& \msgLabel{ans}(\texttt{str})},\ \ses{s}[\role{w_2}] \hasT \type{S_\role{w_2}},\\ 
        &   \ses{s}[\role{s}] \hasT \type{S_\role{s} \| \role c {\o} \msgLabel{wrk}(\role{w_1})},\ \ses{s}[\role{w_1}] \hasT \type{S_\role{w_1} \,\|\, \role c {\o} \msgLabel{ans}(\texttt{str})}
    \end{array}\]}

    \noindent
    Note how reduction is possible because the context split allows the server's parallel type to be extracted into its own context.
    Then, reduction occurs via \Cref{rule:g-cong} and \Cref{rule:g-com1}.
    From here, the context reduces in a similar fashion: first the server communicates a role with the client; followed by the final communication between the client and worker.
    {\normalfont\small\[\begin{array}{r l}
        \redC\ & \\
        &   \ses{s}[\role{c}] \hasT \type{\role{w_1} \& \msgLabel{ans}(\texttt{str})},\ \ses{s}[\role{s}] \hasT \type{S_\role{s} \| \tEnd},\ \ses{s}[\role{w_2}] \hasT \type{S_\role{w_2}},\ 
            \ses{s}[\role{w_1}] \hasT \type{S_\role{w_1} \,\|\, \role c {\o} \msgLabel{ans}(\texttt{str})} \\
        \redC\ & \\ 
        &   \ses{s}[\role{c}] \hasT \type{\tEnd},\ \ses{s}[\role{s}] \hasT \type{S_\role{s} \| \tEnd},\ \ses{s}[\role{w_2}] \hasT \type{S_\role{w_2}},\ 
            \ses{s}[\role{w_1}] \hasT \type{S_\role{w_1}  \| \tEnd} \\
        \type{\equiv}\ & \\
        &   \ses{s}[\role{c}] \hasT \type{\tEnd},\ \ses{s}[\role{s}] \hasT \type{S_\role{s}},\ \ses{s}[\role{w_2}] \hasT \type{S_\role{w_2}},\ 
            \ses{s}[\role{w_1}] \hasT \type{S_\role{w_1}} \\
    \end{array}\]}
\end{example}

\subsubsection{Safety.}

In order to type a session restriction, rule \Cref{rule:t-new} in
\Cref{fig:typing-rules} requires that the session's protocol of the session must
obey a \emph{safety property} $\varphi$. Informally, a safety property requires
that processes exchange values of compatible types and that a sender only ever
selects available branches. Safety is the weakest typing context property required
in order to prove subject reduction.

%
%
\begin{definition}[Safety]\label{def:prop-safe}
    $\varphi$ is a safety property on type environment $\Gamma$ iff:
    {\normalfont\small\begin{mathpar}\mprset{flushleft}
        \infer{S-$\type{\o\&}$}{safe-com}{}
        { 
        \varphi\left(\Gamma \cSplit \ses{s}[\p]: \type{\o_{i \in I}\, \role{\rho_i}\, \m_i (\poly{S_i},\role{\poly{r_i}}) . S_i^\prime} \cSplit \ses{s}[\q]: \type{\p\ \oldAnd_{j \in J}\, \m_j (\poly{S_j^{\prime\prime}},\role{\poly{\alpha_j}}) . S_j^{\prime\prime\prime}}\right)\\\\
        \text{\tab{and} } \exists K \subseteq I \text{ s.t. } \forall k \in K : \role{\rho_k} = \role{q}\\\\
        \text{\tab{implies} } K \subseteq J \text{ and } \forall i \in K : \type{\poly{S_i} \subT \poly{S_i^{\prime\prime}}}  \text{ and } \sizeof{\role{\poly{r_i}}} = \sizeof{\role{\poly{\alpha_i}}}
        }

        \infer{S-$\type{{!}{\o}{\&}}$}{safe-bang}{}
        { 
        \varphi\left(\Gamma \cSplit \ses{s}[\p]: \type{\o_{i \in I}\, \role{\rho_i}\, \m_i (\poly{S_i},\role{\poly{r_i}}) . S_i^\prime} \cSplit \ses{s}[\q]: \type{{!}\role{\rho_0}\ \oldAnd_{j \in J}\, \m_j (\poly{S_j^{\prime\prime}},\role{\poly{\alpha_j}}) . S_j^{\prime\prime\prime}}\right)\\\\
        \text{\tab{and} } \exists K \subseteq I \text{ s.t. } \forall k \in K : \role{\rho_k} = \role{q} \text{\tab{and}\tab} \role{\rho_0} \text{ is either a variable or } \p\\\\
        \text{\tab{implies} } K \subseteq J \text{ and } \forall i \in K : \type{\poly{S_i} \subT \poly{S_i^{\prime\prime}}}  \text{ and } \sizeof{\role{\poly{r_i}}} = \sizeof{\role{\poly{\alpha_i}}}
        }

        \infer{S-$\type{\mu}$}{safe-mu}{}
        {\varphi(\Gamma \cSplit \ses{s}[\q]\hasT \type{\mu \recVar{t}.S}) \tab\text{implies}\tab \varphi(\Gamma \cSplit\ses{s}[\q]\hasT \type{S}\sub{\type{\mu \recVar{t}.S}}{\type{\recVar{t}}})}

        \infer{S-$\role{\alpha}$}{safe-role}{}
        {\varphi(\Gamma) \tab\text{implies}\tab \frv(\Gamma) = \emptyset}

        \infer{S-$\redC$}{safe-r}{}
        {\varphi(\Gamma) \text{ and } \Gamma\ \redC\ \Gammap \tab\text{implies}\tab \varphi(\Gammap)}
    \end{mathpar}}
\end{definition}

A property $\varphi$ is considered \emph{safe} iff it conforms to all conditions specified in \Cref{def:prop-safe}.
Conditions \Cref{rule:safe-com} and \Cref{rule:safe-bang} are concerned with communication.
They state that if two roles in the same session have an output and input type respectively which point at each other, then:
\begin{enumerate*}[label=\textit{(\roman*)}]
    \item they should have at least one common message label;
    \item for each common label, their payloads should be equal in length; and
    \item for each common label, all session types in the payload of the sender should be subtypes of what is expected by the receiver.
\end{enumerate*}

Condition \Cref{rule:safe-mu} requires safety to hold after the unfolding of recursive binders; \Cref{rule:safe-role} requires all role variables used in a context to be bound by that same context; and \Cref{rule:safe-r} requires safety to hold after context reduction.

Users of the type system can re-instantiate $\varphi$ with custom properties (\eg, \emph{termination}), as long as the property used meets the safety requirements.

%% file: metatheory.tex
\subsection{Metatheory}\label{sec:meta}

Unlike most session type theories, \emph{generalised} MPST do not syntactically guarantee any properties on the processes they type.
Rather, they provide a framework for verifying runtime properties on the type context, from which process-level properties may be inferred---the benefit being that these properties are undecidable to check on processes, yet decidable in the realm of the type system.
Furthermore, its generalised nature allows for fine-tuning based on specific requirements of its applications.
Informally, generalisation of the type system works by proving the metatheory parametric of the largest safety property captured by $\varphi$ in \Cref{def:prop-safe}; \ie\ all theorems proved and presented which involve a type context $\Gamma$, assume that $\varphi(\Gamma)$ holds, or ``$\Gamma$ \emph{is safe}''.
Essentially, $\varphi$ describes the minimum (and most general) \emph{safety} requirements made for \emph{subject reduction} to hold. 
This proof technique allows $\varphi$ to be re-instantiated with more specific properties without having to reprove any of the base metatheory.
(We occasionally reference properties other than \emph{safety} in examples.)

\subsection{Subject Reduction and Session Fidelity}

The main results of a generalised MPST system are \emph{subject reduction} (\Cref{thm:sr}) and \emph{session fidelity} (\Cref{thm:sf}).
These theorems allow the type system to be used as a framework for verifying custom properties by re-instantiating $\varphi$.
We note that discussing \emph{how} the generalised type system can be used as a verification framework is \emph{not} the focus of this paper (to this end, we address the interested reader to Scalas and Yoshida~\cite[Section 5]{DBLP:journals/pacmpl/ScalasY19}); 
rather, we build on generalised MPST theory as a means of investigating the expressiveness of replication and first-class roles in MPST.
Hence, the following presents the two theorems---highlighting key differences
with what is standard---but we do not demonstrate the verification of
runtime properties (which is standard). We give the proofs in~\Cref{ap:sr}.

\begin{theorem}[Subject Reduction]\label{thm:sr}
    If $\Theta\cons\Gamma \vdash P$ with $\varphi(\Gamma)$ and $P \redP P^\prime$, then $\exists \Gammap$ s.t. $\Gamma\redC^*\Gammap$ and $\Theta\cons\Gammap \vdash P^\prime$ with $\varphi(\Gammap)$. 
\end{theorem}
%
    %

Intuitively, subject reduction states that any \emph{safe} and \emph{well-typed} process remains safe and well-typed after process reduction.
More formally, the theorem asserts that if a process is typed under a safe context, then the context can match any process reduction to type its continuation whilst retaining safety.
%

\emph{Session fidelity} states that if a context can reduce, then a process it types can observe \emph{at least one} of its reductions.
By virtue of subtyping, a context is allowed to specify paths which need not be followed by a process it types; the key point is that session fidelity requires that there is \emph{at least one} observable reduction.
Using session fidelity, one can prove properties about communication occurring within a single session.
It does not, however, provide grounds for showing such properties on interleaved session communication.
Hence, as is standard in generalised MPST, we define additional assumptions on processes required for session fidelity to hold.

\begin{definition}[Only plays one role]\label{def:one-role}
    (The following is a slight adaptation of \cite[definition 5.3]{DBLP:journals/pacmpl/ScalasY19}.)
    Assuming $\emptyset\cons\Gamma \vdash P$, we say $P$:
    \begin{enumerate}
        \item\label{item:guarded-def} \textbf{has guarded definitions} iff each subterm of P with the form 
        \[\defin{X}{(x_i \hasT \type{S_i})_{i \in 1..n}}{Q\;} P^\prime\]
        we have: \textnormal{$\forall i \in 1..n: \type{S_i \centernot{\subT} \tEnd}$} implies a process call $\recVar{Y}\langle\dots,x_i,\dots\rangle$ can only occur in $Q$ as a subterm of a communication action over channel $x_i$.

        \item\label{item:plays-one-role} \textbf{only plays role $\role p$ in $\ses{s}$, by $\Gamma$} iff 
        \begin{enumerate*}[label=\textit{(\roman*)}]
            \item \Cref{item:guarded-def} holds for P;
            \item $\fv(P)$ = $\emptyset$;
            \item $\Gamma = \Gammab{0},\ses{s}[\p] \hasT \type{U}$ with \textnormal{$\type{U \centernot{\subT} \tEnd}$} and $\pEnd(\Gammab{0})$;
            \item in all subterms $\new{s}{\Gammap} P^\prime$ of $P$, we have $\pEnd(\Gammap)$.
        \end{enumerate*}
    \end{enumerate}
\end{definition}

The purpose of~\Cref{item:guarded-def} is to prevent processes from infinitely
reducing via \Cref{rule:r-proc-call} without communicating,
and~\Cref{item:plays-one-role} identifies a typical application of MPST where a
number of processes $P_\role{p}$ communicate over a multiparty session
$\ses{s}$, with each process playing a \emph{single} role.
The difference in our definition to the standard is that processes $P_\role{p}$
should play a \emph{single} role and not a \emph{unique} role.
This is due to the introduction of replication in our language; note how
reduction with a replicated process is guaranteed to produce multiple processes
playing the same role and is reflected in our definition in condition
\emph{(iii)} of \Cref{item:plays-one-role}, where a channel is mapped to
\emph{runtime type} $\type{U}$, allowing for parallel composition.

Informally, the session fidelity theorem states that, given a safe context that
types a process of a particular structure---\ie\ one governed by the session
fidelity assumptions of \Cref{def:one-role}---then if the context can reduce,
the typed process can match at least one reduction.
Furthermore, after the process matches the context reduction, it remains within
the session fidelity assumption structure.

\begin{theorem}[Session Fidelity]\label{thm:sf}
    Assume $\emptyset\cons\Gamma \vdash P$ with $\varphi(\Gamma)$ and $P \equiv
    {\bigm|_{\p \in I}\! P_\p}$ where each $P_\p$ is either $\0$ (up-to
    $\equiv$), or only plays role $\role p$ in $\ses{s}$. 
    Then, $\Gamma \redC$ implies $\exists \Gammap,P^\prime$ s.t. $\Gamma \redC \Gammap$, $P \redP^* P^\prime$ and $\Gammap \vdash P^\prime$, where $P^\prime \equiv {\bigm|_{\p \in I}\! P_\p^\prime}$ and each $P_\p^\prime$ is either $\0$ (up-to $\equiv$), or only plays role $\role p$ in $\ses{s}$.
\end{theorem}

We now turn our attention to the main focus of this paper, \ie\ exploring the expressiveness and decidability of replication and first-class roles in MPST.

%% file: discussion.tex
\section{Expressivity and Decidability}\label{sec:disc}

This section discusses, and shows by example, the benefits and limitations of \MPSTb.
\Cref{sec:exp} demonstrates the exressivity gained by using replication and first-class roles in MPST,
whilst \Cref{sec:decidability} presents our decidability results.

\subsection{Expressivity}\label{sec:exp}

We begin by demonstrating a common design pattern used for describing protocols, which we call \emph{services}.
We build a number of services to showcase language features, and to describe protocols which---to the best of our knowledge---were previously untypable in any MPST theory.
Specifically, using the increased expressiveness of replication and first-class roles, we define types for \emph{binary trees}, the \emph{dining philosophers problem}, and an \emph{auction}.
Importantly, all examples shown adhere to the decidability requirements discussed later (\cf\ \Cref{sec:decidability}).

\subsubsection{Services.} 

A \emph{service} is a building-block of a protocol, involving some universal receive, with the aim of \emph{offering} a specific interaction.
A \emph{client} interacts with a service to achieve the communication pattern it offers.
Importantly, services may be clients of other services, promoting a modular design of protocols in \MPSTb.

\begin{example}[Ping]\label{ex:ping-service}
    The \emph{ping service} simply responds to a \msgLabel{ping} with a \msgLabel{pong}.
    {\small\[
        \role P \hasT \type{! \role \alpha \& \ping \,.\, \role \alpha {\o} \pong \,.\, \tEnd}
    \]}
\end{example}

A basic yet useful service is given in \Cref{ex:ping-service}, where role $\role P$ offers a \emph{ping} service.
As a convention, we will use capitals for naming services.
We highlight the importance of both replication and free role names in types to be able to design modular components of a protocol---both are integral to designing a sub-protocol agnostic of the larger scope in which it is used.

\subsubsection{Context-Free MPST.}

Context-free session
types~\cite{DBLP:conf/icfp/ThiemannV16,DBLP:journals/iandc/AlmeidaMTV22,DBLP:conf/esop/PocasCMV23}
are a formalism that replace prefix-style session types with individual actions,
along with a sequencing operator \textbf{;} with neutral element \skipy.
The goal of this line of work is to express communication protocols that are not
possible under tail-recursive session types, given their restriction to regular
languages.
The classic example is that of communicating a serialised binary tree.

\begin{example}[Binary tree in standard context-free STs]\label{ex:binary-tree-CF}
    Consider a binary \tree\ data type described by the following context-free grammar.
    {\small\[\tree ::= (\node, \tree, \tree) \midspace \leaf\]}%
    We could attempt to represent a protocol that serialises this data type as
    follows:
    {\small\[\type{\mu \recVar{t}. \o \{\leaf : \skipy,\ \node : \recVar{t}\}}\]}%
    However, this type is not precise enough---it does not guarantee that the correct structure of a binary tree is maintained.
    Work on context-free session types solves this by proposing type systems allowing the following specification:
    {\small\[\type{\mu \recVar{t}. \o \{\leaf : \skipy,\ \node : \skipy;\recVar{t};\recVar{t}\}}\]}%
    Selecting the \type{\node} label now guarantees that \emph{two} sub-trees will follow.
\end{example}

The parallel types presented in \Cref{sec:types}, although not exposed directly to users, lift expressiveness of types in \MPSTb.
In fact, since replicated branches are \emph{permanently available} (by composing continuation types in parallel), we can simulate the sequencing operator \textbf{;} using type-level parallel composition.
%
%

\begin{example}[Binary Tree Service]\label{ex:tree-service}
    Recall the ping service $\role P$ from \Cref{ex:ping-service}.
    We build a \emph{binary tree service} $\role T$ using $\role P$ as shown below:
    {\small\[
        \role T \hasT \type{! \role \beta \& \tree \,.\, \role P {\o} \ping \,.\, {!} \role P \& \pong \,.\, \role \beta \& \left\{\leaf \,.\, \tEnd,\ \node \,.\, \role P {\o} \ping \,.\, \role P {\o} \ping \,.\, \tEnd\right\}}
    \]}%
    The service begins by receiving a request for a \type{\tree} from a client.
    It then sends a \type{\ping} to the ping service, exposing a replicated branch waiting to receive the \type{\pong} reply.
    The client is now free to build the binary tree.
    It is key to note that any \type{\node} sent to the service will subsequently forward \emph{two} \type{\ping} requests to $\role P$.
    In turn, this communication will pull out two copies of the type continuation $\type{\role \beta \& \{\leaf \,.\, \tEnd,\ \node \,.\, \role P {\o} \ping \,.\, \role P {\o} \ping \,.\, \tEnd\}}$, forcing the client to maintain the appropriate binary tree structure.
    For example, if a client $\role t$ wishes to build a tree consisting of one root node and two leaf nodes, its type would be defined as:
    {\small\[
        \role t \hasT \type{\role{T} {\o} \tree \,.\, \role{T} {\o} \node \,.\, \role{T} {\o} \leaf \,.\, \role{T} {\o} \leaf \,.\, \tEnd}
    \]}%
    The metatheoretic framework can be used to determine that any protocol failing to abide by the binary tree structure will result in a \emph{deadlock};
    whilst any safe protocol that obeys the correct structure is \emph{terminating}, \eg\ the protocol $\{\role t, \role T, \role P\}$.
\end{example}


An obscure limitation of the tree service in \Cref{ex:tree-service} is that it can only be used by a single client. 
Consider, for example, two separate clients sending a \type{\node} message to $\role T$.
Since both tree service types communicate with $\role P$ to unroll the replicated branch $\type{! \role P \& \pong}$, the protocol becomes non-deterministic in a non-confluent manner and can result in deadlocked behaviour.
To resolve this issue, we amend the tree service to accept a payload role which should act as a personal ping service for the client; this guarantees that the tree type is only unrolled by the client that made the initial request.
%



\begin{example}[Multi-Client Binary Tree]\label{ex:multi-tree-service}
    We now redesign the binary tree service, this time capable of concurrently building multiple trees for different clients.
    The key difference here being that the new service, $\role{M}$, accepts a role as a payload on the initial request to which it will issue its \type{\ping}s.
    {\small\[\role{M} \hasT \type{{!} \role \alpha \& \tree(\role \beta) \,.\, \role \beta {\o} \ping \,.\, {!}\role \beta \& \pong \,.\, \role \alpha \& \left\{
            \leaf \,.\, \tEnd,
            \node \,.\, \role \beta {\o} \ping \,.\, \role \beta {\o} \ping \,.\, \tEnd
    \right\}}\]}%
    %
    %
    {\small\[\type{S_\p} = \type{! \role{M} \& \ping \,.\, \role{M} {\o} \pong \,.\, \tEnd}\]}%
    Multiple clients can now issue concurrent requests to the tree service whilst maintaining safety.
    A sample (terminating) protocol is that of $\{\role{t_1}, \role{t_2}, \role{p_1}, \role{p_2}, \role M\}$, where $\role{p_1},\role{p_2} \hasT \type{S_\p}$, and the types for $\role{t_1}, \role{t_2}$ are given by:
    {\small\begin{align*}
        \role{t_1} \hasT& \type{\role{M} {\o} \tree(\role{p_1}) \,.\, \role{M} {\o} \node \,.\, \role{M} {\o} \leaf \,.\, \role{M} {\o} \leaf \,.\, \tEnd}\\
        \role{t_2} \hasT& \type{\role{M} {\o} \tree(\role{p_2}) \,.\, \role{M} {\o} \node \,.\, \role{M} {\o} \leaf \,.\, \role{M} {\o} \node \,.\, \role{M} {\o} \leaf \,.\, \role{M} {\o} \leaf \,.\, \tEnd}
    \end{align*}}
\end{example}

\subsubsection{Replication vs.\ Recursion.}

We have seen that replication and parallel composition increases the expressive power of MPST beyond that of tail-recursion.
Naturally, one might ask, ``\textit{is recursion still needed?}''
We find replication and recursion in MPST to be mutually non-inclusive---\ie\ both can produce protocols which \emph{cannot} be typed under the other construct.
We have already demonstrated this in one direction with the binary tree examples; below we showcase how recursion cannot be replaced by replication.

\begin{example}[Lock Service]\label{ex:lock-service}
    The lock service provides clients with a \emph{mutex lock}. 
    {\small\[\role{L} \hasT \type{! \role{\theta} \& \lock \,.\, \mu \recVar{t} .\, \role{\theta} \& \left\{\acquire \,.\, \role \theta \& \release \,.\, \recVar{t},\ \done \,.\, \tEnd\right\}}\]}%
    When a client requests a lock from $\role L$, a copy of the recursive continuation is exposed.
    The recursive definition allows sequences of \type{\acquire} and \type{\release} messages to be received.
    It is key to note that, whilst replication maintains a top-level branch that is permanently available to receive a message, the top-level action in a recursive definition is \emph{not} fixed. 
\end{example}

Copies of the continuation type of a replicated receive are executed concurrently.
\Cref{ex:lock-service} provides a service for roles to enter race-sensitive portions of a protocol, as if it were an atomic action.
We demonstrate its use by typing the dining philosophers problem.

\begin{example}[Dining Philosophers]\label{ex:dining-philosophers}
    A number of philosophers gather to eat on a round table. 
    Each plate is separated by a single chopstick, and a philosopher needs two chopsticks to eat.
    The dining philosophers problem requires the philosophers to employ an algorithm to ensure the table does not get deadlocked.
    In such a setting, we can view \emph{chopsticks as services} and \emph{philosophers as clients}.
    Assuming a size of $n$-philosophers, we define the type for a chopstick as:
    {\small\[(\role{C_i})_{1..n} \hasT \type{\role L {\o} \lock \,.\, {!} \role \eta \& \left\{
        \begin{array}{l}
            \take \,.\, \role L {\o} \acquire \,.\, \role \eta {\o} \ok \,.\, \role \eta \& \give \,.\, \role L {\o} \release \,.\, \tEnd \\
            \done \,.\, \role L {\o} \done \,.\, \tEnd
        \end{array}
    \right\}}\]}%
    Before offering its service, a chopstick requests a lock from $\role L$.
    This ensures that every chopstick has its own lock that it may \type{\acquire} and \type{\release}.
    The lock is used to guarantee that a chopstick is only ever taken by a single philosopher at a time.
    A chopstick then waits for a \type{\take} request from a philosopher; receiving one will result in it attempting to \type{\acquire} the lock.
    This acquisition is only successful if the same role has not already requested it in some other parallel composition.
    If the lock was already acquired, then the $\type{\role L {\o} \acquire}$ will block until the lock is released.
    Acquiring the lock sends an $\type{\msgLabel{ok}}$ back to the philosopher,
    symbolising that they have successfully obtained the chopstick.
    When done from eating, the philosopher may then send back a \type{\give} message, which in turn releases the lock, as the chopstick is now available to be taken by a different philosopher.
    
    We can now write an algorithm for philosophers.
    First, a naive approach:
    {\small\begin{align*}
        (\role{p_i})_{i \in 1..n} :&\, \type{\role{C_i} {\o} \take .\, \role{C_{i\smallplus1}} {\o} \take .\, \role{C_i} \& \ok .\, \role{C_{i\smallplus1}} \& \ok .\, \role{C_i} {\o} \give .\, \role{C_{i\smallplus1}} {\o} \give .\, \q {\o} \fin .\, \tEnd}\\
        \role{q} :&\, \type{\role{p_1} \& \fin .\, \cdots .\, \role{p_n} \& \fin .\, \role{C_1} {\o} \done .\, \cdots .\, \role{C_n} {\o} \done .\, \tEnd}
    \end{align*}}
    
    \noindent
    Every philosopher $\role{p_i}$ has a similar type.
    They begin by requesting to take the chopsticks to their left and right---note that this results in every chopstick receiving two \type{\take} requests.
    Receiving both \type{\ok} messages means the philosopher can eat, and subsequently \type{\give} back the chopsticks.
    Finally, when finished, philosophers send a \type{\fin} to role $\q$, acting as a clean-up for the protocol.
    The protocol $\{\role{p_i}, \role{q}, \role{C_i}, \role{L}\}_{i\in 1..n}$ is
    safe, but fails typechecking for $\varphi = \textit{terminating}$.
    In fact, the na\"ive protocol allows for scenarios in which all philosophers take a single chopstick, resulting in a deadlock.
    This problem has many solutions; we present the simplest in which philosophers take turns to eat.
    (Key changes are underlined.)
    {\small\begin{align*}
        \type{S_1} &\,\smallequals\, \type{\role{C_1} {\o} \take .\, \role{C_2} {\o} \take .\, \role{C_1} \& \ok .\, \role{C_2} \& \ok .\, \role{C_1} {\o} \give .\, \role{C_2} {\o} \give .\, \underline{\role{p_2} {\o} \fin} .\, \tEnd}\\
        \type{S_2} &\,\smallequals\, \type{\underline{\role{p_{i\smallminus1}} \& \fin} . \role{C_i} {\o} \take . \role{C_{i\smallplus1}} {\o} \take . \role{C_i} \& \ok . \role{C_{i\smallplus1}} \& \ok . \role{C_i} {\o} \give . \role{C_{i\smallplus1}} {\o} \give . \underline{\role{p_{i\smallplus1}} {\o} \fin} . \tEnd}\\
        \type{S_3} &\,\smallequals\, \type{\underline{\role{p_{n\smallminus1}} \& \fin} . \role{C_n} {\o} \take . \role{C_1} {\o} \take . \role{C_n} \& \ok . \role{C_1} \& \ok . \role{C_n} {\o} \give . \role{C_1} {\o} \give . \role{q} {\o} \fin . \tEnd}\\
        \role{q^\prime} &\hasT\, \type{\role{p_n} \& \fin .\, \role{C_1} {\o} \done .\, \cdots .\, \role{C_n} {\o} \done .\, \tEnd}
    \end{align*}}%
    Here, all philosophers other than the first must wait for the previous to \type{\fin}ish eating before they can request to take their chopsticks.
    The updated protocol $\{\type{\role{p_1} \hasT S_1}, \role{p_i} \hasT \type{S_2}, \role{p_n} \hasT \type{S_3}, \role{q^\prime}, \role{C_j}, \role{L}\}^{i \in 2..n\smallminus1}_{j\in 1..n}$ now typechecks for $\varphi =$ \emph{terminating}.
\end{example}

The previous examples demonstrate how recursion hidden by a universal receive can be used to mimic changes in state.
Our final example does the inverse, \ie\ we show how a universal receive hidden by a recursive binder can be used to model resources which eventually reach some permanent state.
In addition, we show that universal receives model \emph{fair races}, since they do not impose an order on how communication is handled.

\begin{example}[Auction]\label{ex:barter-store}
    A merchant $\role{m}$ sets up an auction $\role{A}$ to accept bids from some buyers $\role{b}$.
    A merchant can employ different mechanisms for choosing who to sell to (\eg, first come first served, highest bid, biased selling, \etc); but must always respond to buyers with either a \type{\yes}, \type{\no}, or \type{\notavail} message.
    {\small\[\type{{\role{A}}} \hasT \type{{!}\role{\alpha} \& \bid(\intT) .\, \role{m} {\o} \bid(\intT, \role \alpha).\, \tEnd}\]
    \[\role{m} \hasT \type{\mu \recVar{t} .\, \role{A} \& \bid(\intT, \role{\beta}) .\, \role{\beta} {\o} \left\{
        \begin{array}{l}
            \yes .\, {!} \role{A} \& \bid(\intT, \role \kappa) .\, \role \kappa {\o} \notavail .\, \tEnd \\
            \no  .\, \recVar{t}
        \end{array}
    \right\}}\]
    \[(\role{b_i})_{i \in 1..n} \hasT \type{\mu \recVar{t} .\, \role{A} {\o} \bid(\intT).\, \role{m} \& \left\{\yes.\, \tEnd,\ \no.\, \recVar{t},\ \notavail.\, \tEnd\right\}}\]}
    Buyers $\role{b_i}$ race to send \type{\bid}s to the auction service.
    In turn, the auction forwards bids and buyer role identifiers to the merchant, who processes bids sequentially (but still in an arbitrary order).
    If the merchant declines a bid, then the client is offered another chance; if the merchant accepts a bid, then it exposes a replicated receive which informs any further buyers that the product is no longer available.
    It is key to note that, unlike in \Cref{ex:dining-philosophers} where we used locks to avoid race conditions, races here are not only allowed but are integral to the protocol.
    Additionally, by uncovering a replicated receive, the merchant enters a \emph{permanent} state.
    These two characteristics guarantee that, no matter the selling algorithm employed by the merchant: 
    \emph{(i)} bids always arrive in a fair arbitrary order; and
    \emph{(ii)} the product can only be sold once.
\end{example}

\subsubsection{Discussion.}

We have now shown that \emph{replication} and \emph{recursion} are mutually
non-inclusive, and that our extension increases the expressiveness of 
MPST.
It is important to understand the dependencies between added features and 
the expressiveness gained.
Since \MPSTb\ is a \emph{conservative} extension, it is guaranteed that 
the increase of expressiveness derives from our two extensions: 
\emph{1.} the addition of \emph{replication}; and
\emph{2.} the addition of \emph{first-class roles}.

Replication alone is enough to increase the expressiveness of MPSTs w.r.t. 
the Chomsky hierarchy. 
We note that, \eg, \Cref{ex:tree-service} could be easily re-written 
without the universal receive, especially since it should not be used 
by multiple clients to uphold deadlock-freedom---thus, 
\emph{replication in MPSTs increases their expressiveness to that of 
context-free languages}.

First-class roles in our formalism refers to: 
\emph{(i)} \emph{universal receives} acting as binders on role variables; and
\emph{(ii)} the ability to \emph{pass roles} as payloads in messages.
Universal receives allow \emph{protocols to be designed agnostic of the client pool} 
(consider \Cref{ex:lb-intro,ex:ping-service,ex:multi-tree-service,ex:lock-service});
and also act as a \emph{fair way of introducing races}---\eg\ \Cref{ex:lb-intro} 
describes a load balancer that responds to requests in \emph{no particular order}.
Role passing allows for \emph{safe distributed choice}. 
In a load balancer (in general), it is \emph{impossible} for a client to know 
which worker will service its request. 
In \Cref{ex:lb-intro}, role passing allows the server to inform clients of its choice, 
and also informs the worker of the identity of the client.
Role passing increases expressiveness by introducing dependencies into a protocol. 
For instance, \Cref{ex:barter-store} uses role passing to ensure the merchant correctly 
services the right buyers; without it, a merchant could not respond to requests without bias.

As a final note, first-class roles are different to, \eg, \emph{delegation} or \emph{multiple sessions} (both supported in \MPSTb) since they act \emph{inside} a session. 
Therefore our system can still be used to check for properties such as deadlock-freedom, 
which is not possible with interleaved sessions without other mechanisms such as an 
interaction typing system~\cite{DBLP:conf/concur/BettiniCDLDY08} or 
priorities~\cite{DBLP:conf/fossacs/DardhaG18}.

\subsection{Decidability}\label{sec:decidability}

As one may expect, the added expressiveness of replication and first-class roles into types does not come without a cost.
Unlike the base theory that this work extends, even though our language models synchronous communication, typechecking may be \emph{undecidable} in the general case.
In the following, we discuss decidability of typechecking in detail. 
We show that typechecking is only as decidable as the safety property; we provide examples of types that make typechecking problematic; and we provide strategies for determining whether a protocol is captured by a decidable subset of \MPSTb.

\begin{theorem}[Decidablility of type checking]\label{thm:dec-gen}
    If $\varphi$ is decidable, then typechecking is decidable.
\end{theorem}

\begin{proof}
   Since typing rules in \Cref{fig:typing-rules} can be deterministically applied
   based on the structure of a process $P$, and a typing context need only be
   split a finite number of times to separate all linear types,
   there are a finite number of contexts that can be tried for each rule that requires a context split. 
    %
    Lastly, subtyping is decidable~\cite{DBLP:journals/acta/GayH05} (decidability of subtyping replicated types is equivalent to regular branching types, and of parallel types is equivalent to checking multiple session types); and $\varphi$ is decidable by assumption. \qed
\end{proof}

\Cref{thm:dec-gen} states that decidability of type checking is only as decidable as property $\varphi$.
In \Cref{ex:infinite-beh}, we will demonstrate why $\varphi$ may not necessarily be {decidable} in the general case for the type semantics presented in \Cref{fig:type-sem}.
To do this, we first define \emph{behavioural sets} of type contexts (as in \cite[appendix K]{DBLP:journals/pacmpl/ScalasY19tech}).

\begin{definition}[Behavioural set]\label{def:beh}
    The behavioural set of a type context, written $\beh(\Gamma)$, is given by $\beh(\Gamma) = \unf^*(\{\Gammap \| \Gamma \redC^* \Gammap\})$; where $\unf^*$ is the closure of $\unf$---a function that unfolds all top-level recursive binders in a set of contexts.
    (Full definitions of $\unf$ and $\unf^*$ are standard and given in \Cref{app:dec}.)
\end{definition}

Informally, the behavioural set of a context $\Gamma$ is the set of \emph{(i)} its reductions; and \emph{(ii)} its reductions' unfoldings.
The benefit of $\beh$ is that it mechanically abides by conditions \Cref{rule:safe-mu} and \Cref{rule:safe-r} from \Cref{def:prop-safe}.
Therefore, to determine whether $\beh(\Gamma)$ is a safety property, all that is required is to exhaustively check the contexts that inhabit $\beh(\Gamma)$ against the remaining conditions of \Cref{def:prop-safe}.

\begin{example}\label{ex:beh}
    Consider a context $\Gamma = \{\ses{s}[\p] \hasT \type{\mu \recVar{t} . \q {\o} \m . \recVar{t}}, \ses{s}[\q] \hasT \type{\mu \recVar{t$^\prime$} . \p \& \m . \recVar{t$^\prime$}}\}$.
    The behavioural set of $\Gamma$ is given by:
    \[\beh(\Gamma) = \left\{\left\{
        \begin{array}{l}
            \ses{s}[\p] \hasT \type{\mu \recVar{t} . \q {\o} \m . \recVar{t}}, \\
            \ses{s}[\q] \hasT \type{\mu \recVar{t$^\prime$} . \p \& \m . \recVar{t$^\prime$}}
        \end{array}
    \right\},\left\{
        \begin{array}{l}
            \ses{s}[\p] \hasT \type{\q {\o} \m . \mu \recVar{t} . \q {\o} \m . \recVar{t}}, \\
            \ses{s}[\q] \hasT \type{\p \& \m . \mu \recVar{t$^\prime$} . \p \& \m . \recVar{t$^\prime$}}
        \end{array}
    \right\}\right\}\]
    Notice that the left element is the original context after $0$ reduction
    steps, whereas the right element is the unfolding of $\Gamma$.
    Moreover, any further reductions only yield contexts (and unfoldings) already captured by these two elements.
\end{example}

The next example context is problematic for typechecking.

\begin{example}\label{ex:infinite-beh}
    Consider a context $\Gamma = \{\ses{s}[\p] \hasT \type{\mu \recVar{t} . \q {\o} \m . \recVar{t}}, \ses{s}[\q] \hasT \type{! \p \& \m . \role r {\o} \m}\}$.
    The behavioural set of $\Gamma$ is given by:
    \[\beh(\Gamma) = \left\{
    \begin{array}{l}
        \left\{
        \begin{array}{l}
            \ses{s}[\p] \hasT \type{\mu \recVar{t} . \q {\o} \m . \recVar{t}}, \\
            \ses{s}[\q] \hasT \type{! \p \& \m . \role r {\o} \m}
        \end{array}
        \right\},\left\{
            \begin{array}{l}
                \ses{s}[\p] \hasT \type{\q {\o} \m . \mu \recVar{t} . \q {\o} \m . \recVar{t}}, \\
                \ses{s}[\q] \hasT \type{! \p \& \m . \role r {\o} \m }
            \end{array}
        \right\}, 
        \\
        \left\{
            \begin{array}{l}
                \ses{s}[\p] \hasT \type{\mu \recVar{t} . \q {\o} \m . \recVar{t}}, \\
                \ses{s}[\q] \hasT \type{! \p \& \m . \role r {\o} \m \| \role r {\o} \m}
            \end{array}
        \right\},\left\{
            \begin{array}{l}
                \ses{s}[\p] \hasT \type{\q {\o} \m . \mu \recVar{t} . \q {\o} \m . \recVar{t}}, \\
                \ses{s}[\q] \hasT \type{! \p \& \m . \role r {\o} \m \| \role r {\o} \m}
            \end{array}
        \right\}, \cdots
    \end{array}
    \right\}\]
    Indeed, $\beh(\Gamma)$ is \emph{infinite}.
    This is a result of how replication and parallel composition are modelled in \Cref{fig:type-sem}.
    In fact, the type semantics for replicated communication allows for context reduction to yield \emph{larger} types.
    Note how in this example, the contexts that inhabit $\beh(\Gamma)$ get infinitely larger by pulling out infinitely many copies of type $\type{\role r {\o} \m}$.

    Furthermore, we point out that infinite behavioural sets are not only a result of recursive communication with replicated branches.
    Consider, \eg\ a $\Gammap$ = $\{\ses{s}[\role p] \hasT \type{!\role \alpha \& \m . \role \alpha {\o} \m^\prime . \role r {\o} \m} ,
    \ses{s}[\role q] \hasT \type{\p {\o} \m . !\role \beta \& \m^\prime . \role \beta {\o} \m}\}$.
    Such a context will also pull out infinitely many copies of type $\type{\role r {\o} \m}$, because the replicated communication forms an infinite loop.
    Lastly, it is key to note that $\beh(\Gammapp)$ is finite for any $\Gammapp$ that does not contain replicated branches, since there is no other way for a context reduction to yield a larger type.
\end{example}

Knowing whether $\beh(\Gamma)$ is (in-)finite is key for our main decidability result.

\begin{theorem}[Decidability of $\beh$]\label{thm:dec-beh}
    Let $\varphi = \beh(\Gamma)$. If $\beh(\Gamma)$ is finite, then $\varphi$ is decidable.
\end{theorem}
\begin{proof}
    Since $\beh(\Gamma)$ contains all possible reductums and unfoldings of
    $\Gamma$, then conditions \Cref{rule:safe-r} and \Cref{rule:safe-mu} are
    satisfied immediately.
    Therefore, to determine whether $\beh(\Gamma)$ is a safety property, we may
    exhaustively check all inhabitants of $\beh(\Gamma)$ against conditions
\Cref{rule:safe-com}, \Cref{rule:safe-bang}, \Cref{rule:safe-role}, which is
decidable since $\beh(\Gamma)$ is finite (by assumption); and since subtyping
and $\frv$ are decidable. \qed \end{proof}

\Cref{thm:dec-beh} states that $\varphi$ is \emph{decidable} for any $\varphi = \beh(\Gamma)$ where $\beh(\Gamma)$ is a finite set.
In other words, if a protocol can be shown to have a finite behavioural set, then typechecking for that protocol is \emph{decidable}.
This could be done manually for each protocol; 
however, to further increase the practicality of our type system, 
we present two strategies for restricting protocols into a subset 
of \MPSTb\ with finite behavioural sets.

\subsubsection{Decidability Strategies.}

The strategies we present for restricting protocols to decidable subsets of \MPSTb\ all follow a similar blueprint.
Essentially, we wish to establish properties on $\Gamma$ with decidable
approximations that imply that $\beh(\Gamma)$ is finite. 
Then, by Theorems~\ref{thm:dec-gen} and~\ref{thm:dec-beh} we obtain decidable typechecking.

The following defines, and gives examples, of each strategy; then, we show these strategies are sound and discuss how they can be approximated.

\Cref{def:triv-fin} prevents types like $\Gamma$ in \Cref{ex:infinite-beh} using
a na\"ive approach; put simply, $\tf$ captures protocols where all clients of a 
replicated server are intrinsically non-recursive and non-replicated.

\begin{definition}[Trivially finite]\label{def:triv-fin}
    A context $\Gamma$ is trivially finite, $\tf(\Gamma)$, iff:
    \begin{enumerate}
        \item\label{item:tf-2} no type in the body of a recursive binder sends to a replicated branch; and
        \item\label{item:tf-1} continuations of replicated branches do not send to other replicated branches.
    \end{enumerate}
\end{definition}


\begin{example}
    The protocols modelling the dining philosophers problem in \Cref{ex:dining-philosophers} are \emph{trivially finite}.
    Note how the chopstick services make the initial request to the lock service \emph{before} they offer their replicated branch.
\end{example}

For other protocols we need a more nuanced strategy.
\Cref{def:cf} formalises ``loops'' in a protocol which may result from replicated 
servers infinitely bouncing messages amongst each other (such as $\Gammap$ in \Cref{ex:infinite-beh}).

%
%

\begin{definition}[Loop free]\label{def:cf}
    Given a protocol $\Psi$, and a context $\Gamma$ derived from $\Psi$ (possibly after a number of reductions), 
    a \textbf{cycle} in the LTS of $\Gamma$ is defined as the series of transitions s.t., 
    for $\Gamma = \Gammap\cSplit\ses{s}[\p] \hasT \type{!\role \rho \&_{i\in I} \m_i (\poly{T_i}).S_i}$
    \[
        \Gamma
        \redS{\aCom{s}{\q}{\p}{\m_k}}\!\!
        \left(\!\redS{\aCom{s}{\role{p^\prime_j}}{\role{q^\prime_j}}{\m^{\prime}_j}}\!\right)_{j \in 1..n}\!\!\!\!\!\!\!\!\!\!\!\!\!\!\!
        \redS{\aCom{s}{\q}{\p}{\m_k}}
        \Gammapp
    \]
    where $k \in I$, and for any $\role{p^\prime}, \role{q^\prime}, \m^\prime, n, \Gammapp$.
    A cyclic replicated communication path (\textbf{CRCP}) is defined as a cycle with these added conditions:
    \begin{enumerate}
        \item $\Gamma(\ses{s}[\q]) = \type{S_\q}$ s.t. either $ \type{S_\q {\equiv}\, !S^\& \| U}$ or 
              $\type{S_\q}$ appears after a recursive binder in $\Psi(\q)$, for any $\type{S^\&},\type{U}$; and
        \item $\forall x \in 1..n : $ $\Gamma \redS{\aCom{s}{\q}{\p}{\m_k}} \left(\!\redS{\aCom{s}{\role{p^\prime_l}}{\role{q^\prime_l}}{\m^{\prime}_l}}\!\right)_{l \in 1..x-1}\!\!\!\!\!\!\!\!\!\!\!\!\!\!\!\!
              \type{\Gamma^{\prime\prime\prime}} \redS{\aCom{s}{\role{p^\prime_x}}{\role{q^\prime_x}}{\m^{\prime}_x}}$ 
              we have: \\
              $\type{\Gamma^{\prime\prime\prime}}(\ses{s}[\role{q^\prime_x}]) = \type{S_\role{q^\prime_x}}$ s.t. either $\type{S_\role{q^\prime_x} {\equiv}\, !S^\& \| U}$ or $\type{S_\role{q^\prime_x}}$ appears after a recursive binder in $\Psi(\role{q^\prime_x})$.
%
    \end{enumerate}
    We say $\Gamma$ is \textbf{loop free}, written $\lf(\Gamma)$, iff the LTS of $\Gamma$ 
    does not contain a CRCP.
\end{definition}

Essentially, a \emph{cycle} in the LTS of a context $\Gamma$: 
\emph{(i)} starts with an incoming communication action into a replicated type; 
\emph{(ii)} performs some intermediary transitions; and
\emph{(iii)} ends with the transition that began the cycle.
A CRCP is a special case of a cycle, where all intermediary transitions must also be 
between roles that have a replicated type; or form part of the body of a recursive type.
Finally, a context is \emph{loop free} iff its LTS does not produce any CRCPs.

\begin{example}\label{ex:crcp}
    Contexts $\Gamma$ and $\Gammap$ from \Cref{ex:infinite-beh} contain CRCPs: 
    $\Gamma$ contains a CRCP at $\q$ with $0$ intermediary transitions; and
    $\Gammap$ contains two CRCPs, at $\q$ and $\p$, both with $1$ intermediary transition
    forming part of a replicated type.
\end{example}

\begin{example}\label{ex:cf}
    The protocols in \Cref{ex:lb-intro,ex:tree-service,ex:multi-tree-service} are \emph{loop free}: \Cref{ex:lb-intro} because there are no cycles; \Cref{ex:tree-service,ex:multi-tree-service} because the cycle between the \msgLabel{pong} branch on $\role{T}$ (resp. $\role M$) and the \msgLabel{ping} branch on $\role P$ (resp. $\role{p_1},\role{p_2}$) includes communication with the (non-replicated/-recursive) client; breaking the CRCP.
\end{example}

\begin{proposition}\label{prop:ioc}
    If $\beh(\Gamma)$ is infinite, then $\Gamma$ contains a CRCP.
\end{proposition}

\begin{proof}
    From the type semantics (\Cref{fig:type-sem}), we observe that the only reductions that can yield \emph{larger} types are communications with replicated branches.
    Therefore, it follows that for $\beh(\Gamma)$ to be infinite, there must be some 
    reoccurring transitions in the LTS of $\Gamma$ that repeatedly communicates with 
    a replicated branch---\ie\ there is a communication action on a replicated branch, 
    followed by any number of intermediary transitions which then end with the initial 
    communication action on the replicated branch; where all intermediary transitions 
    must be non-finite.
    This is the definition of a CRCP (\Cref{def:cf}). \qed
\end{proof}

\begin{theorem}[Strategy soundness]\label{thm:strat-sound}
    Given a context $\Gamma$, $\Phi(\Gamma)$ implies $\beh(\Gamma)$ is finite, for $\Phi \in \{\tf,\lf\}$
\end{theorem}

\begin{proof}
    \emph{Case} \tf. By contradiction: assume $\beh(\Gamma)$ is infinite, then, by \Cref{prop:ioc}, $\Gamma$ contains a CRCP;
    but, by \Cref{def:cf}, a CRCP will violate at least one of the conditions for \tf\ in \Cref{def:triv-fin}---contradiction.
    \\
    \emph{Case} \lf. By contradiction: assume $\beh(\Gamma)$ is infinite, then, by \Cref{prop:ioc}, $\Gamma$ contains a CRCP---contradiction; 
    therefore $\beh(\Gamma)$ is \emph{finite}.
    \qed
\end{proof}

\noindent
\emph{Approximations.} 
Properties \tf\ and \lf\ are both \emph{decidable} for all protocols \emph{without} 
first-class roles:
\tf\ can be determined via a linear traversal of a type context; and 
\lf\ can be checked by constructing a directed graph of visited replicated 
branches in a context, then checking that the graph is acyclic (which is 
decidable).
An approximation is only required for protocols using role variables, since their 
values can only be known at runtime.
This approximation would treat any role variable in a selection type to have the 
capability of reaching \emph{any other role}.

\begin{example}[Approximation false negative]\label{ex:approx}
    Consider the following protocol:
    \[
    \p \hasT \type{!\role{\alpha} \& \m . \role \alpha {\o} \m^\prime }
    \tab\tab
    \q \hasT \type{\p {\o} \m . \p \& \m^\prime . \role r {\o} \m^\prime . \role r \& \m}
    \tab\tab
    \role r \hasT \type{!\role{\beta} \& \m^\prime . \role \beta {\o} \m}
    \]%
    Although the above is \emph{trivially finite} ($\p$ and $\role r$ do not communicate), 
    it would be \emph{falsely} flagged as \emph{not} \tf\ because $\role \alpha$ is over-approximated to include $\role r$.
\end{example}

It is key to note that false negatives of the approximation are avoidable by 
\emph{requiring unique branching labels on replicated types}.
Furthermore, even with the approximation, all examples presented in this paper (except \Cref{ex:barter-store}) are captured by either \lf\ or \tf.
With respect to~\Cref{ex:barter-store}, the presented protocol still yields a
finite behavioural set, and thus by \Cref{thm:dec-beh} and \Cref{thm:dec-gen},
typechecking it is decidable.
We aim to continue exploring further strategies (especially ones which can capture protocols such as \Cref{ex:barter-store}) in future work. 


%% file: related.tex
\section{Related Work}\label{sec:related}
The \MAGpi calculus~\cite{DBLP:conf/esop/BrunD23} makes use of generalised MPST
theory in order to type \emph{failure-prone} communications
(i.e., message loss,
reordering, and delays). Key to their approach is the use of timeouts to
detect and handle message loss; as with related approaches (e.g., Barwell et
al.~\cite{DBLP:conf/concur/BarwellSY022}), this often means that session types
are made more complex in order to handle each potential failure point.
Our approach is most closely related to that of Le Brun \&
Dardha~\cite{DBLP:conf/forte/BrunD24}, who introduce \MAGpib as a modification
of \MAGpi to incorporate type-level replication: this has the advantage of
simplifying client-server interactions by only requiring \emph{clients} to
handle potential failures.
However, the aims of our work and that of Le Brun \&
Dardha~\cite{DBLP:conf/forte/BrunD24} are significantly different: while their
aim is specifically to use replication as a methodology to simplify failure
handling, our work is a more fundamental study of the consequences of
type-level replication on expressiveness and decidability.
In particular we make use of a more standard base MPST calculus (i.e., a
calculus that does not include undirected receives, nor rules that
model failures and message reordering), and we make use of \emph{synchronous}
communication semantics. Nevertheless our calculus allows for more
interesting use of replication: unlike \MAGpib we allow \emph{nested
replication} and \emph{recursion}, whereas \MAGpib only allows replication at
the top-level and processes must be finite.
Furthermore, as a result of our inclusion of \emph{first-class roles}, as well
as using replication \emph{in tandem with} recursion, we can type 
protocols that make non-trivial use of mutual exclusion and races,
all of which would be \emph{inexpressible in \MAGpib.}

Toninho \& Yoshida~\cite{DBLP:journals/toplas/ToninhoY18} assess the relative
expressiveness of a multiparty session calculus and a process calculus inspired
by classical linear logic, showing that MPST calculi allow strictly more
expressive process networks (i.e., those that can include cycles).  As part of
this investigation they explore a limited form of type-level replication that
permits a liveness property.  However, their system does not consider first-class
roles and pre-dates generalised MPST so is guided primarily by global types, and
is therefore less expressive than $\MPSTb$.

Replicated session types have been used to a limited extent in a
wide variety of works on binary session types
(e.g.,~\cite{DBLP:journals/mscs/CairesPT16,DBLP:conf/fossacs/DardhaG18,DBLP:conf/esop/CairesP17,DBLP:journals/jfp/Wadler14}),
primarily in works pertaining to Curry-Howard interpretations of propositions as session types, where the exponential modality from linear
logic $!A$ is typically linked to replication from the $\pi$-calculus.
Several further lines of work investigate client-server communication following this
correspondence.
Kokke et al.~\cite{KokkeMW20} investigate an extension of the logically-inspired
HCP calculus~\cite{KokkeMP19} with two dual modalities $!_{n}A$ and $?_{n}A$ to
type a pool of $n$ clients and a replicated server that can service $n$ requests
respectively, and show that their calculus allows nondeterministic behaviour
while still preventing deadlocks and ensuring termination.
Qian et al.~\cite{DBLP:journals/pacmpl/QianKB21} develop \textsf{CSLL}
(client-server linear logic) that uses the dual \emph{coexponential} modalities
\textexclamdown$A$ and \textquestiondown$A$ to type servers and client pools
respectively, along with rules to merge client pools.
The subtle difference between the \textexclamdown$A$ modality and the
exponential $!A$ being that the former (informally) serves type $A$ \emph{only as many
times as required according to client requests}.
This is similar to how our type system operates, given that replicated receives
only pull out copies of continuations upon communication.
Multiple requests induce non-determinism into further reductions; 
in our work this is observed through parallel types, which in the work of Qian et
al.~\cite{DBLP:journals/pacmpl/QianKB21} is observed through
hyperenvironments~\cite{KokkeMP19,FowlerKDLM23}.
%
%
Unlike all these works, we focus on \emph{multiparty} session types, 
where interacting with a replicated channel 
spawns a process that \emph{remains in the same session}.
This results in  our key novelty, \ie\ our account of replication in the
\emph{type semantics} with the use of parallel types.

Marshall \& Orchard~\cite{DBLP:journals/corr/abs-2203-12875} investigate the
effects of adding a \emph{semiring graded necessity} modality (a generalisation
of the $!$ modality) to a session-typed $\lambda$-calculus,
showing interesting consequences such as replicated servers and multicast
communication. 
We posit that the two systems can
type different protocols: 
while \MPSTb cannot straightforwardly encode multicast
communication, it is difficult to see how their approach would scale to the
examples we describe in~\Cref{sec:disc}.

Rocha \& Caires~\cite{RochaC23} introduce CLASS, a process calculus with a
correspondence to Differential Linear Logic~\cite{Ehrhard18}. CLASS integrates
session-typed communication, reference cells with mutual exclusion, and
replication. Their calculus guarantees preservation and progress, the proof of
the latter property requiring a logical relation. CLASS can encode the dining
philosophers problem, making essential use of shared state; in contrast our
implementation relies on the interplay between replication and recursion.

Deni{\'{e}}lou et al.~\cite{DBLP:journals/corr/abs-1208-6483} introduce 
parameterised MPST as a means of designing protocols for parallel 
algorithms. 
Their formalism allows for parameterisation of participants in the 
form of $\role{client}[i]$, representing the $i^{\text{th}}$ client from some 
group of $n$ clients, for a bound $n$.
%
%
The key difference between this formalism and \MPSTb\ is that our 
approach preserves, and allows for the fair handling of, \emph{races}.
Parameterised MPST enforce a predetermined prioritisation on the order 
of communication (thus, \Cref{ex:barter-store} cannot be expressed in 
that system).

%% file: conc.tex
\section{Conclusion}\label{sec:conc}

We presented \MPSTb, a conservative extension of the standard
multiparty session $\pi$-calculus which introduces for the first time
\emph{replication} and \emph{first-class roles}, and proved its metatheory.
We have shown that the interplay between replication and recursion
allows us to describe interesting and previously inexpressible MPST protocols
such as those that rely on races and mutual exclusion, as well as giving a
method by which we can express context-free protocols.
Although replication can have implications for decidability of typechecking, we
have identified sufficient conditions that can determine decidability and
provided syntactic approximations for decidability.
For future work, it would be interesting to investigate an extension of \MPSTb\
with polymorphism, as this would improve on the modular design of protocols
already promoted by the type system.
Furthermore, we wish to continue exploring the decidability of typechecking 
to find more general approximations.

%% file: cong.tex
\section{Structural Congruence}\label{app:cong}

\begin{definition}[Structural Congruence]\label{def:cong}
    {\small\begin{mathpar}
    \inferrule{}{P\|Q \equiv Q\|P}
    
    \inferrule{}{(P\|Q)\|R \equiv P\|(Q\|R)} 

    \inferrule{}{P\|\0 \equiv P}

    \inferrule{}{\newnt{s}\,\0 \equiv \0} 
    
    \inferrule{}{\newnt{s}\newnt{s^\prime}\, P \equiv \newnt{s^\prime}\newnt{s}\, P }

    \inferrule{\ses{s}\not\in \fs(P)}{\newnt{s}\,(P\|Q) \equiv P\|\newnt{s}\,Q}

    \inferrule{}{\defin{X}{\poly{x}}{P\;} \0 \equiv \0}

    \inferrule{\ses{s} \not\in \fs(P)}{\defin{X}{\poly{x}}{P\;} \newnt{s}\,Q \equiv \newnt{s}\,(\defin{X}{\poly{x}}{P\;} Q)}

    \inferrule
    {\dpv(D)\cap\fpv(Q) = \emptyset}
    {\definD{D\;} (P\|Q) \equiv (\definD{D\;}P)\| Q}

    \inferrule
    {(\dpv(D)\cup\fpv(D))\cap\dpv(D^\prime) = (\dpv(D^\prime)\cup\fpv(D^\prime))\cap\dpv(D) = \emptyset}
    {\definD{D\;} (\definD{D^\prime\;} P) \equiv \definD{D^\prime\;}(\definD{D\;} P)}
    \end{mathpar}}
\end{definition}

%% file: prelims.tex
\section{Lemmata}\label{app:lemmata}

\subsection{Context Operations}

\begin{lemma}
    $\Gamma = \Gammab{1} \cSplit \Gammab{2} \ \textit{implies} \  \Gammab{1} + \Gammab{2} = \Gamma$.
\end{lemma}

\begin{proof}
    Holds by induction on the derivation of $\Gamma = \Gammab{1} \cSplit \Gammab{2}$.
    Notably, in the cases where $\Gamma = \Gammap,c\hasT\type{U}$, we observe that context composition `$,$' implies that $c \not\in \dom(\Gammap)$.
    Thus, we may directly apply the context addition rules to obtain the original context. \qed
\end{proof}

\subsection{Substitution}

\begin{lemma}\label{lem:split-subs}
    If $\Gamma = \Gammab{1} \cSplit \Gammab{2}$, then $\Gamma\sub{\role q}{\role \alpha} = \Gammab{1}\sub{\role q}{\role \alpha}  \cSplit \Gammab{2}\sub{\role q}{\role \alpha} $.
\end{lemma}

\begin{proof}
    By induction on the derivation of $\Gamma = \Gammab{1} \cSplit \Gammab{2}$.

    Case $\emptyset$: $\emptyset = \emptyset \cSplit \emptyset \implies \emptyset\sub{\role q}{\role \alpha} = \emptyset\sub{\role q}{\role \alpha} \cSplit \emptyset\sub{\role q}{\role \alpha}$.

    Case \emph{cL}: $\Gamma = \Gammap,c:\type{U}$ and $\inferrule
    {\Gammap = \Gammab{1} \cSplit \Gammab{2}}
    {\Gammap,c \hasT \type{U} = \Gammab{1},c \hasT \type{U} \cSplit \Gammab{2}}$.
    We observe that:

    \begin{flalign}
        &(\Gammap, c \hasT \type{U})\sub{\role q}{\role \alpha} = \Gammap\sub{\role q}{\role \alpha}, c \hasT \type{U}\sub{\role q}{\role \alpha} = \Gammap\sub{\role q}{\role \alpha} \cSplit c \hasT \type{U}\sub{\role q}{\role \alpha} && \text{(by \frv\ and def of $\cSplit$)}\label{lem2-ln4} \\
        &\Gammap\sub{\role q}{\role \alpha} = \Gammab{1}\sub{\role q}{\role \alpha} \cSplit \Gammab{2}\sub{\role q}{\role \alpha} && \text{(by the ind. hyp.)} \label{lem2-ln5} \\
        &c\not\in \dom(\Gammap),\ \therefore c\not\in \dom(\Gammab{1}) && \text{(from hyp.)} \label{lem2-ln6}\\
        &(\Gammap, c \hasT \type{U})\sub{\role q}{\role \alpha} = \Gammab{1}\sub{\role q}{\role \alpha} \cSplit \Gammab{2}\sub{\role q}{\role \alpha} \cSplit c \hasT \type{U}\sub{\role q}{\role \alpha} && \text{(by (\ref{lem2-ln4}),(\ref{lem2-ln5}))} \\
        & = \Gammab{1}\sub{\role q}{\role \alpha} \cSplit c \hasT \type{U}\sub{\role q}{\role \alpha} \cSplit \Gammab{2}\sub{\role q}{\role \alpha} && \text{(since $\cSplit$ commutative)} \\
        & = \Gammab{1}\sub{\role q}{\role \alpha}, c \hasT \type{U}\sub{\role q}{\role \alpha} \cSplit \Gammab{2}\sub{\role q}{\role \alpha} && \text{(by (\ref{lem2-ln6}))}
    \end{flalign}

    Cases \emph{cR} and $\|$ are similar to case \emph{cL}.

    Case $\role \alpha$: Follows directly from the ind. hyp. since role substitution does not affect singletons. \qed
\end{proof}

\begin{lemma}\label{lem:end-subs}
    If \textnormal{$\pEnd(\Gamma)$}, then $\Gamma\sub{\q}{\role \alpha} = \Gamma$.
\end{lemma}

\begin{proof}
    By rule induction on $\pEnd(\Gamma)$, observing that for any context which is end-typed, there are no elements that can be substituted. \qed
\end{proof}

\begin{lemma}\label{lem:role-subs-val}
    Assume $\Gamma \vdash V \hasT \type{T}$. 
    Then, $\Gamma \sub{\q}{\role \alpha} \vdash V\sub{\q}{\role \alpha} \hasT \type{T}\sub{\q}{\role \alpha}$.
\end{lemma}
\begin{proof}
    The lemma holds trivially if $\role \alpha$ does not occur free in $\type{T}$.
    If $\role \alpha$ is free in $\type{T}$, then there are two cases to consider.
    First, when $\type T = \role \alpha$. \\
    From the assumption and rule \cref{rule:t-roleVar} we obtain $\role \alpha \hasT \role \alpha \vdash \role \alpha \hasT \role \alpha$.
    We must show that $\role \alpha \hasT \role \alpha \vdash \role q \hasT \role q$, which holds by rules \cref{rule:t-weak} and \cref{rule:t-roleVal}. \\
    The second case is when $\type T = \St$ for some $\St$ in which $\role \alpha$ occurs free.
    From assumption and inversion of rule \cref{rule:t-sub}, we know that $\Gamma = c \hasT \type{S^\prime}$, $V = c$, and $\type{S^\prime \subT S}$.
    Hence, we must prove that $(c \hasT \type{S^\prime})\sub{\q}{\role \alpha} \vdash c \hasT \type{S} \sub{\q}{\role \alpha}$.
    By the definition of \frv, this is the same as proving $c \hasT (\type{S^\prime}\sub{\q}{\role \alpha}) \vdash c \hasT \type{S} \sub{\q}{\role \alpha}$.
    In turn, by rule \cref{rule:t-sub}, we must prove that $\type{S^\prime \subT S}$ implies $\type{S^\prime}\sub{\q}{\role \alpha} \subT \type{S} \sub{\q}{\role \alpha}$, which holds by \cref{lem:subt-preserves-subs}. \qed
\end{proof}
\begin{lemma}\label{lem:role-subs-proc}
    Assume $\Theta \cons \Gamma \vdash P$. 
    Then, $\Theta \cons \Gamma \sub{\q}{\role \alpha} \vdash P\sub{\q}{\role \alpha}$.
\end{lemma}
\begin{proof}
    By induction on the derivation of $\Theta \cons \Gamma \vdash P$. 

    Case \cref{rule:t-0}: Follows from \cref{lem:end-subs}.

    Case \cref{rule:t-par}: Follows from the ind. hyp. and \cref{lem:split-subs}.

    Case \cref{rule:t-choice} follows directly from the ind. hyp..

    Case \cref{rule:t-new}: R.T.P. $\Theta \cons \Gammap \vdash \new{s}{\type{\Psi}} P^\prime$ implies $\Theta \cons \Gammap\sub{\q}{\role \alpha} \vdash \new{s}{\type{\Psi}} P^\prime\sub{\q}{\role \alpha}$.
    We know:
    \begin{flalign}
        &\inferrule{\varphi(\extract_\ses{s}(\Psi)) \\\\ \textnormal{ $\varphi$ is a \emph{safety} property} \\\\ \ses{s} \not\in \Gammap \\ \Theta \cons \Gammap + \extract_\ses{s}(\Psi) \vdash P^\prime}{\Theta \cons \Gammap \vdash \new{s}{\type{\Psi}} P^\prime} 
        && \text{(from assumption)} \label{lem5-ln10}\\
        & \Theta \cons (\Gammap + \extract_\ses{s}(\Psi))\sub{\q}{\role \alpha} \vdash P^\prime\sub{\q}{\role \alpha} && \text{(ind. hyp.)} \label{lem5-ln11}
    \end{flalign}
    But since $\varphi$ is a \emph{safety} property and $\varphi(\extract_\ses{s}(\Psi))$, then by condition \cref{rule:safe-role}, we know that there are \emph{no free role variables in} $\extract_\ses{s}(\Psi)$. Hence we infer the below:
    \begin{flalign}
        & \Theta \cons \Gammap\sub{\q}{\role \alpha} + \extract_\ses{s}(\Psi) \vdash P^\prime\sub{\q}{\role \alpha} && \text{(by (\ref{lem5-ln10}), (\ref{lem5-ln11}) and \cref{rule:safe-role})} \label{lem5-ln12} \\
        & \Theta \cons \Gammap\sub{\q}{\role \alpha} \vdash \new{s}{\type{\Psi}} P^\prime\sub{\q}{\role \alpha} && \text{(by (\ref{lem5-ln10}), (\ref{lem5-ln12}) and \cref{rule:t-new})}
    \end{flalign}

    Case \cref{rule:t-def}: This case follow directly from the ind. hyp., but we make a note that this is only the case because of how we define role substitution on process definition, \ie\ 
    \[
        (\defin{X}{\poly{x \hasT \type{S}}}{P\;} Q) \sub{\q}{\role \alpha} := \defin{X}{\poly{x \hasT \type{S}\sub{\q}{\role \alpha}}}{P\sub{\q}{\role \alpha}\;} Q\sub{\q}{\role \alpha}
    \]

    Cases for \cref{rule:t-bang,rule:t-rcv,rule:t-send,rule:t-call} are similar and all follow from some application of \cref{lem:role-subs-val}. We expand on case \cref{rule:t-send} below:

    \noindent
    Let $\Gamma = \Gammap \cSplit \Gammab{\o} \cSplit \Gammab{r} (\cSplit \Gammab{i})_{i \in 1..n}$.\\
    R.T.P. $\Theta \cons \Gamma \cSplit \type{\Gamma_\o} \cSplit \type{\Gamma_r} (\!\cSplit \type{\Gamma_i})_{i\in 1..n}
    \vdash 
    c[\role{\rho}] \o \m \langle (V_i)_{i \in 1..n} \rangle \,.\, P^\prime$ implies\\ 
    $\Theta \cons (\Gammap \cSplit \type{\Gamma_\o} \cSplit \type{\Gamma_r} (\!\cSplit \type{\Gamma_i})_{i\in 1..n})\sub{\q}{\role \alpha}
    \vdash 
    (c[\role{\rho}] \o \m \langle (V_i)_{i \in 1..n} \rangle \,.\, P^\prime)\sub{\q}{\role \alpha}$.

    We know:
    \begin{flalign}
        &\label{lem5-ln14}
        \inferrule
        {
        \type{\Gamma_\o} \vdash c \hasT \type{\role{\rho} {\o} \m (\poly{T}) . S^\prime} \\
        \type{\Gamma_r} \vdash \role{\rho} \hasT \role \rho \\\\
        (\type{\Gamma_i} \vdash V_i \hasT \type{T_i})_{i \in 1..n} \\
        \Theta \cons \Gammap + c \hasT \type{S^\prime}  \vdash P^\prime
        }
        {\Theta \cons \Gammap \cSplit \type{\Gamma_\o} \cSplit \type{\Gamma_r} (\!\cSplit \type{\Gamma_i})_{i\in 1..n}
        \vdash 
        c[\role{\rho}] \o \m \langle (V_i)_{i \in 1..n} \rangle \,.\, P^\prime}   && \text{(from assumption)}\\
        &\label{lem5-ln15} 
        \type{\Gamma_\o}\sub{\q}{\role \alpha} \vdash c \hasT \type{\role{\rho} {\o} \m (\poly{T}) . S^\prime}\sub{\q}{\role \alpha} && \text{(by (\ref{lem5-ln14}) and \cref{lem:role-subs-val})} \\
        & \type{\Gamma_r}\sub{\q}{\role \alpha} \vdash (\role{\rho} \hasT \role \rho)\sub{\q}{\role \alpha} && \text{(by (\ref{lem5-ln14}) and \cref{lem:role-subs-val})} \\
        & \forall i \in 1..n : \type{\Gamma_i}\sub{\q}{\role \alpha} \vdash V_i\sub{\q}{\role \alpha} \hasT \type{T_i} \sub{\q}{\role \alpha} && \text{(by (\ref{lem5-ln14}) and \cref{lem:role-subs-val})} \\
        &\label{lem5-ln18} 
         \Theta \cons (\Gammap + c \hasT \type{S^\prime})\sub{\q}{\role \alpha} \vdash P^\prime\sub{\q}{\role \alpha} && \text{(by (\ref{lem5-ln14}) and ind. hyp.)}
    \end{flalign}
    $\therefore$ from lines (\ref{lem5-ln15})--(\ref{lem5-ln18}) and rule \cref{rule:t-send}, we obtain the thesis. \qed

\end{proof}

\begin{lemma}\label{lem:chan-subs}
    Assume $\Theta \cons \Gamma \cSplit x \hasT \St \vdash P$ and $\Gammap \vdash \ses{s}[\p] \hasT \type{S^{\prime\prime}}$ where $\type{S^{\prime\prime} \subT S}$. 
    Then, $\Theta \cons \Gamma + \Gammap \vdash P\sub{\ses{s}[\p]}{x}$.
\end{lemma}

\begin{proof}
    Similar to \cite[Lemma B.1]{DBLP:journals/pacmpl/ScalasY19tech}; in our case we make use of the context operations instead of directly using context compositionssss.
\end{proof}

\subsection{Subtyping}

\begin{lemma}\label{lem:subt-preserves-subs}
    $\type{S \subT S^\prime}$ implies $\type{S}\sub{\q}{\role \alpha} \subT \type{S^\prime} \sub{\q}{\role \alpha}$.
\end{lemma}

\begin{proof}
    By coinduction on the derivation of $\type{S \subT S^\prime}$.
    We expand on the branching case below; 
    other cases follow similar reasoning.

    Consider the case where $\St = \type{\role{\rho} \&_{i \in I}\m_i (\poly{T_i}) . S_i}$ and $\type{S^\prime} = \type{\role{\rho} \&_{i \in I \cup J}\m_i (\poly{T_i^\prime}) . S_i^\prime}$. We know:
    \begin{flalign}
        &\inferrule
        {\type{(\poly{T_i} \,\subT\, \poly{T_i^\prime})_{i \in I}} \\ \type{(S_i^{\prime\prime}\,\subT\, S_i^{\prime\prime\prime})_{i \in I}}}
        {\type{\role{\rho} \&_{i \in I}\m_i (\poly{T_i}) . S_i^{\prime\prime} \,\subT\, \role{\rho} \&_{i \in I \cup J}\m_i (\poly{T_i^\prime}) . S_i^{\prime\prime\prime}}}
        && \text{(by assumption)}
    \end{flalign}
    If $\role \rho = \role \alpha$, then the substitution takes place on both sides, and thus subtyping is preserved top-level.
    By the coinductive hyp. we know that the property holds for $\type{(S_i^{\prime\prime}\,\subT\, S_i^{\prime\prime\prime})_{i \in I}}$, and for all payload types that are session types.
    For any payload types that are \emph{not} session types, we know by definition of subtyping (\cref{def:subtyping}) that $\type{\poly{T_i}} = \type{\poly{T_i^\prime}}$.
    Then, if $\type{\poly{T_i}} = \role \alpha$, the substitution takes place on both sides, preserving the subtype relation.  \qed
\end{proof}

\begin{lemma}\label{lem:end-subt}
    $\pEnd(\Gamma)$ implies if $\Gammap \subT \Gamma $, then $ \Gammap = \Gamma$.
\end{lemma}

\begin{proof}
    By induction on $\pEnd(\Gamma)$, follows immediately from \cref{def:end}.
\end{proof}

\begin{lemma}\label{lem:split-subt}
    If $\Gamma = \Gammab{1} \cSplit\Gammab{2}$ and $\Gammap \subT \Gamma$, then $\Gammap = \Gammab{3} \cSplit \Gammab{4}$ where $\Gammab{3} \subT \Gammab{1}$ and $\Gammab{4} \subT \Gammab{2}$.
\end{lemma}

\begin{proof}
    By induction on $\Gamma = \Gammab{1} \cSplit\Gammab{2}$, observing that subtyping does not affect the domain of a context, hence subtyping over splits can be preserved.
\end{proof}

\begin{lemma}\label{lem:add-subt}
    If $\Gammab{1} + \Gammab{2} = \Gamma$ and $\Gammap \subT \Gamma$, then $\exists \Gammab{3} $ and $\Gammab{4}$ s.t. $\Gammab{3} + \Gammab{4} = \Gammap$ and $\Gammab{3} \subT \Gammab{1}$ and $\Gammab{4} \subT \Gammab{2}$.
\end{lemma}

\begin{proof}
    Similar to \cref{lem:split-subt}, this time by induction on $\Gammab{1} + \Gammab{2} = \Gamma$.
\end{proof}

\begin{lemma}\label{lem:narrowing-val}
    Assume $\Gamma \vdash V \hasT \type T$ and $\Gammap \subT \Gamma$.
    Then, $\Gammap \vdash V \hasT \type T$.
\end{lemma}

\begin{proof}
    By induction on the typing derivation.

    Case \cref{rule:t-roleVal}: $\emptyset \vdash \q \hasT \q$ and only $\emptyset \subT \emptyset$.

    Case \cref{rule:t-roleVar}: $\role \alpha \hasT \role \alpha \vdash \role \alpha \hasT \role \alpha$ but only $\role \alpha \hasT \role \alpha \subT \role \alpha \hasT \role \alpha$.

    Case \cref{rule:t-sub}: $\inferrule{\type{S}\subT\type{S^\prime}}{c \hasT \St \vdash c \hasT \type{S^\prime}}$. 
    Consider $c \hasT \type{S^{\prime\prime}} \subT c \hasT \type{S}$, then by transitivity of subtyping, $\inferrule{\type{S^{\prime\prime}}\subT\type{S^\prime}}{c \hasT \type{S^{\prime\prime}} \vdash c \hasT \type{S^\prime}}$.

    Case \cref{rule:t-weak}: $\inferrule{\Gammab{1} + \Gammab{2} = \Gammap \\\\ \Gammab{1} \vdash V \hasT \type{T} \\ \pEnd(\Gammab{2}) }
    {\Gammap \vdash V \hasT \type{T}}$.
    Consider $\Gammapp \subT \Gammap$, then by \cref{lem:add-subt}, $\type{\Gammab{1}^\prime} + \type{\Gammab{2}^\prime} = \Gammapp$ and $\type{\Gammab{1}^\prime} \subT \Gammab{1}$ and $\type{\Gammab{2}^\prime} \subT \Gammab{2}$.
    By the ind. hyp. we know $\type{\Gammab{1}^\prime} \vdash V \hasT \type{T}$; and by \cref{lem:end-subt} we know $\pEnd(\type{\Gammab{2}^\prime})$. 
    Therefore, we conclude by $\inferrule*[right={T-Wkn}]{\type{\Gammab{1}^\prime} + \type{\Gammab{2}^\prime} = \Gammapp \\\\ \type{\Gammab{1}^\prime} \vdash V \hasT \type{T} \\ \pEnd(\type{\Gammab{2}^\prime}) }
    {\Gammapp \vdash V \hasT \type{T}}$ \qed
\end{proof}

\begin{lemma}\label{lem:narrowing-proc}
    Assume $\Theta \cons \Gamma \vdash P$ and $\Gammap \subT \Gamma$.
    Then, $\Theta \cons \Gammap \vdash P$.
\end{lemma}

\begin{proof}
    By induction on the typing derivation. Follows from \cref{lem:narrowing-val}. \qed
\end{proof}

\subsection{Congruence}



\begin{lemma}\label{lem:subject-cong}
    Assume $\Theta \cons \Gamma \vdash P$ and $P \equiv P^\prime$.
    Then, $\exists \Gammap$ s.t. $\type{\Gamma \equiv \Gammap}$ and $\Theta \cons \Gammap \vdash P^\prime$.
\end{lemma}

\begin{proof}[Sketch]
    By cases on the definition of $P \equiv P^\prime$, observing that either a matching congruence can be applied to the type context to preserve the judgement, or the process remains typable under the original context.
\end{proof}

\subsection{Safety}

\begin{lemma}
    If $\Gamma,\Gammap$ is {safe}, then $\Gamma$ is {safe}.
\end{lemma}

\begin{proof}
    By contradiction. Assume $\Gamma$ not safe. 
    Then, by \cref{def:prop-safe} condition \cref{rule:safe-r}, 
    there is a $\Gammapp$ s.t. $\Gamma \redC^*\Gammapp$ and $\Gammapp$ either violates \cref{rule:safe-com} or \cref{rule:safe-role} (possibly after some applications of \cref{rule:safe-mu}).
    But by \cref{rule:g-cong}, $\Gamma,\Gammap = \Gamma \cSplit \Gammap \redC^* \Gammapp + \Gammap$ that is not safe.
    Therefore, $\Gamma,\Gammap$ violates \cref{rule:safe-r} and is itself not safe---contradiction.
    Hence, $\Gamma,\Gammap$ is {safe}.\qed
\end{proof}


\begin{lemma}\label{lem:safe-interact}
    If $\Gamma$ is safe and $\Gamma = \ses{s}[\p] \hasT \type{S_{\o}},\ses{s}[\q] \hasT \type{S_{\&}}\; \subT\; \Gammap$ and $\Gammap \redC$, then $\Gamma \redC$.
\end{lemma}

\begin{proof}[Sketch]
    We observe that by virtue of subtyping (\cref{def:subtyping}), at least one overlapping message that allows for reduction in the supertype will remain present in the subtype.
\end{proof}

\begin{lemma}\label{lem:red-commutes-subt}
    Assume $\Gamma$ is {safe}, $\Gamma \subT \Gammap$, and $\Gammap \;\redC\; \type{\Gamma^{\prime\prime}}$.
    Then, $\exists \type{\Gamma^{\prime\prime\prime}}$ s.t. $\Gamma \;\redC\; \type{\Gamma^{\prime\prime\prime}}$ with $\type{\Gamma^{\prime\prime\prime}} \subT \type{\Gamma^{\prime\prime}}$.
\end{lemma}

\begin{proof}
    If $\Gammap \redC \Gammapp$, then the reduction is a result from the communication between some two entries $\ses{s}[\p] \hasT \type{S_\o^\prime},\ses{s}[\q] \hasT \type{S_\&^\prime}$ in $\Gammap$.

    Since $\Gamma \subT \Gammap$, then $\Gamma = \Gammab{0} \cSplit \Gammab{\role \rho},\ses{s}[\p] \hasT \type{S_\o},\ses{s}[\q] \hasT \type{S_\&}$, where $\Gamma_\role{\rho}$ only contains singletons (and hence $\pEnd(\Gammab{\role \rho})$).

    By \cref{lem:safe-interact}, we know $\Gamma \redC \type{\Gamma^{\prime\prime\prime}}$. Lastly, we obtain $\type{\Gamma^{\prime\prime\prime}} \subT \Gammapp$ since 
    \emph{(i)} for entries $\ses{s}[\p]$ and $\ses{s}[\q]$, we observe from \cref{def:subtyping} that type continuations preserve subtyping;
    and \emph{(ii)} the remainder of the context is unaffected by the communication. \qed
\end{proof}

\subsection{Session Inversion}

\begin{lemma}
    If $P \equiv \0$ and $\Theta \cons \Gamma \vdash P$, then $\pEnd(\Gamma)$.
\end{lemma}

\begin{proof}
    By \cref{lem:subject-cong} $\exists \Gammap \type{\equiv} \Gamma$ s.t. $\Theta \cons \Gammap \vdash \0$.
    By \cref{rule:t-0}, $\pEnd(\Gammap)$.
    By \cref{def:end}, we obtain $\pEnd(\Gamma)$.
\end{proof}

\begin{lemma}
    Assume $\emptyset \cons \Gamma\vdash\ \bigm\vert_{\q \in I} P_\q$ with each $P_\q$ either being $\0$ (up-to $\equiv$) or only playing role $\q$ in $\ses{s}$.
    Then, $\Gamma = \Gammab{0}, \left\{\ses{s}[\role{q}] \hasT \type{U_\q} \right\}_{\q \in I^\prime}$ (for some $I^\prime$) and with $\pEnd(\Gammab{0})$. 
    Moreover, $\forall \q \in I^\prime$:
    \begin{enumerate}
        \item if $\type{!\role \rho \&_{j \in J} \m_j (\poly{T_j}).S_j} \;\subT\; \type{U_\q}$ then $\q \in I$ and for some $\rc,\rc^\prime,$ and $J^\prime \supseteq J$, either:
        \begin{itemize}
            \item $P_k \equiv \rc\left[!\ses{s}[\role{\q}][\role \rho]\&_{j\in J^\prime}\m_j(\poly{b_j}).P^\prime_j\right]$ or
            \item \textnormal{$P_k \equiv \rc\left[\begin{array}{c}
                \defin{X}{x_1 \hasT \type{S^{\prime}_1},\dots,x_n \hasT\type{S^{\prime}_n}}{\rc^\prime \left[!x_l[\role \rho]\&_{j \in J^\prime}\m_j(\poly{b_j}).P^\prime_j\right]} \\
                \recVar{X}\left\langle \ses{s^\prime_1}[\role{r_1}],\dots,\ses{s^\prime_{l-1}}[\role{r_{l-1}}], \ses{s}[\role{\q}],\ses{s^\prime_{l+1}}[\role{r_{l+1}}],\dots,\ses{s^\prime_n}[\role{r_n}] \right\rangle
            \end{array}\right]$} with $1 \leq l \leq n$;
        \end{itemize}
        \item if $\type{\role \rho \&_{j \in J} \m_j (\poly{T_j}).S_j} \;\subT\; \type{U_\q}$ then $\q \in I$ and for some $\rc,\rc^\prime,$ and $J^\prime \supseteq J$, either:
        \begin{itemize}
            \item $P_k \equiv \rc\left[!\ses{s}[\role{\q}][\role \rho]\&_{j\in J^\prime}\m_j(\poly{b_j}).P^\prime_j\right]$ or
            \item \textnormal{$P_k \equiv \rc\left[\begin{array}{c}
                \defin{X}{x_1 \hasT \type{S^{\prime}_1},\dots,x_n \hasT\type{S^{\prime}_n}}{\rc^\prime \left[!x_l[\role \rho]\&_{j \in J^\prime}\m_j(\poly{b_j}).P^\prime_j\right]} \\
                \recVar{X}\left\langle \ses{s^\prime_1}[\role{r_1}],\dots,\ses{s^\prime_{l-1}}[\role{r_{l-1}}], \ses{s}[\role{\q}],\ses{s^\prime_{l+1}}[\role{r_{l+1}}],\dots,\ses{s^\prime_n}[\role{r_n}] \right\rangle
            \end{array}\right]$} with $1 \leq l \leq n$;
        \end{itemize}
        \item if $\type{ \o_{j \in J}\role{\rho_j} \m_j (\poly{T_j}).S^{\prime}_j} \;\subT\; \type{U_\q}$ then $\q \in I$ and for some $\rc,\rc^\prime$ and $J^\prime \subseteq J$, either:
        \begin{itemize}
            \item $P_k \equiv \rc\left[ {\textstyle\sum_{j \in J^\prime}} \ses{s}[\role{q}][\role{\rho_j}]\o \m_j\langle\poly{V_j}\rangle.P^\prime_j\right]$ or
            \item \textnormal{$P_k \equiv \rc\left[\!\begin{array}{c}
                \defin{X}{x_1 \hasT \type{S^{\prime}_1},\dots,x_n \hasT\type{S^{\prime}_n}}{\rc^\prime \left[{\textstyle\sum_{j \in J^\prime}}x_l[\role{\rho_j}]{\o} \m_j\langle\poly{V_j}\rangle . P^\prime_j\right]}\!\! \\
                \recVar{X}\left\langle \ses{s^\prime_1}[\role{r_1}],\dots,\ses{s^\prime_{l-1}}[\role{r_{l-1}}], \ses{s}[\role{q}],\ses{s^\prime_{l+1}}[\role{r_{l+1}}],\dots,\ses{s^\prime_n}[\role{r_n}] \right\rangle
            \end{array}\right]$} with $1 \leq l \leq n$.
        \end{itemize}
        \item if $\type{ S_1 \| \cdots \| S_n} \;\subT\; \type{U_\q}$ then $\q \in I$ and $P_\q \equiv P_1 \| \cdots \| P_n$ s.t. $\forall j \in 1..n: \emptyset \cons \Gammab{0},\{\ses{s}[\q] \hasT \type{S_j}\} \vdash P_j$, with each $P_j$ either being $\0$ (up-to $\equiv$) or only plays $\q$ in $\ses{s}$.
    \end{enumerate}
    Further,
    \begin{enumerate*}[start=5]
        \item $\forall \q \in I \setminus I^\prime : P_\q \equiv \0$.
    \end{enumerate*} 
\end{lemma}

\begin{proof}
    Items \emph{1--3} and \emph{5} follow similar proofs to the base version of this theory, found in \cite[Theorem B.4]{DBLP:journals/pacmpl/ScalasY19tech}.
    (Main differences being that for \emph{3} we infer the shape of the choice before the selection; and the proof for \emph{1} is new but almost identical to that of \emph{2}.)
    Item \emph{4} follows from the definition of \cref{def:one-role}.
\end{proof}

%% file: sr.tex
\section{Subject Reduction and Session Fidelity}\label{ap:sr}

\begin{customthm}{1}
    If $\Theta \cons \Gamma \vdash P$ with $\Gamma$ safe.
    Then, $P \redP P^\prime$ implies $\exists \Gammap$ safe s.t. $\Gamma\ \redCstar\ \Gammap$ and $\Theta \cons \Gammap \vdash P^\prime$.
\end{customthm}
\begin{proof}
    The proof is by induction on the derivation of $P \redP P^\prime$.
    Cases \Cref{rule:r-c}, \Cref{rule:r-bang1}, \Cref{rule:r-bang2} are similar.
    For each of these we infer the shape of $\Gamma$ by inversion of typing rules, observing that $\Gamma$ reduces via communication actions, and rebuilding a typing derivation for the reduced process using the reduced context.
    Cases \Cref{rule:r-choice} and \Cref{rule:r-proc-call} are straight forward.
    Case \Cref{rule:r-cong} holds from subject congruence (\cref{lem:subject-cong}), and case \Cref{rule:r-ctx} holds by a further proof by induction on the structure of $\rc$.

    We demonstrate the proof with case \Cref{rule:r-bang2}:

    \noindent
    We may assume:
    \begin{enumerate*}[label=\textbf{(A\arabic*)}]
        \item $\Theta \cons \Gamma \vdash P$ \label{ass:sr1}\tab
        \item $\varphi(\Gamma)$ \label{ass:sr2}\tab
        \item $P \redP P^\prime$ \label{ass:sr3}
    \end{enumerate*}
    
    From the hypothesis and \cref{rule:r-bang2}, we infer the shape of $P$ and $P^\prime$:
    \begin{flalign}
        P &= \snd{\ses{s}[\role{q}]}{\role{p}}{\m_k\langle \poly{d} \rangle} Q\ \|\ {!}\rcvI{\ses{s}[\role{p}]}{\role{\alpha}}{\m_i(\poly{b_i})} Q_i^\prime 
        & \reason{by rule \cref{rule:r-bang2}}{sr1} \\
        P^\prime &= Q\ \|\ {!}\rcvI{\ses{s}[\role{p}]}{\role{\alpha}}{\m_i(\poly{b_i})} Q_i^\prime\ \|\ Q_k^\prime \sub{\poly{d}}{\poly{b_k}}\sub{\q}{\role \alpha} 
        & \reason{by rule \cref{rule:r-bang2}}{sr2} \\
        k &\in I & \reason{by rule \cref{rule:r-bang2}}{sr3}
    \end{flalign}

    We can now infer the shape of $\Gamma$ through inversion of typing rules:
    {\small\begin{mathpar}
        \inferrule*[right={\cref{rule:t-par}}]
        {
            \inferrule
            {}
            { \Gammab{L} = \Gammab{L_0} \cSplit \Gammab{\o} \cSplit \Gammab{\role \rho} (\cSplit \Gammab{L_i})_{i \in 1..n} \vdash \snd{\ses{s}[\role{q}]}{\role{p}}{\m_k\langle d_1,\dots,d_n \rangle} Q }
        \\
            \inferrule
            {}
            { \Gammab{R} = \Gammab{R_0} \cSplit \Gammab{!} \vdash {!}\rcvI{\ses{s}[\role{p}]}{\role{\alpha}}{\m_i(\poly{b_i})} Q_i^\prime }
        }
        { \Gamma = \Gammab{L} \cSplit \Gammab{R} \vdash \snd{\ses{s}[\role{q}]}{\role{p}}{\m_k\langle \poly{d} \rangle} Q\ \|\ {!}\rcvI{\ses{s}[\role{p}]}{\role{\alpha}}{\m_i(\poly{b_i})} Q_i^\prime }

        \inferrule*[right={\cref{rule:t-send}}]
        { \Gammab{\o} \vdash \ses{s}[\q] \hasT \type{\p \o m_k (T_1,\dots,T_n).S_L} 
        \\
        \Gammab{\role \rho} \vdash \p \hasT \p
        \\
        \forall i \in 1..n : \Gammab{L_i} \vdash d_i \hasT \type{T_i}
        \\
        \Gammab{L_0} + \ses{s}[\q] \hasT \type{S_L} \vdash Q
        }
        { \Gammab{L_0} \cSplit \Gammab{\o} \cSplit \Gammab{\role \rho} (\cSplit \Gammab{L_i})_{i \in 1..n} \vdash \snd{\ses{s}[\role{q}]}{\role{p}}{\m_k\langle d_1,\dots,d_n \rangle} Q }

        \inferrule*[right={\cref{rule:t-bang}}]
            { \Gammab{!} \vdash \ses{s}[\p] \hasT \type{!\role \alpha \&_{i\in I} m_i (\poly{T^\prime}).S_{R_i}}
            \\
              \pEnd(\Gammab{R_0})
            \\
              \left(\Gammab{R_0} + \ses{s}[\p] \hasT \type{S_{R_i}} + \poly{b_i} \hasT \type{\poly{T_i}} \insnew{\role \alpha} \vdash Q_i^\prime\right)_{i \in I}
            }
            { \Gammab{R_0} \cSplit \Gammab{!} \vdash {!}\rcvI{\ses{s}[\role{p}]}{\role{\alpha}}{\m_i(\poly{b_i})} Q_i^\prime }
    \end{mathpar}}

    Since $\Gamma$ is \emph{safe}, we now that:
    \begin{flalign}
        &\|\type{\poly{T}}\| = \|\type{\poly{T^\prime}}\| && \reason{by condition \cref{rule:safe-bang}}{sr4}\\
        & \forall j \in 1..\|\type{\poly{T}}\| : \type{T_j} \subT \type{T_j^\prime}\ \text{ if $\type{T_j}$ is a session type} && \reason{by condition \cref{rule:safe-bang}}{sr5}
    \end{flalign}

    Without loss of generality, we assume $\type{\poly{T}}$ and $\type{\poly{T^\prime}}$ to be ordered as session types first, followed by role singletons, s.t. 
    $\type{\poly{T}} = \type{\poly{S^\prime},\poly{\role{r}}}$ and $\type{\poly{T^\prime}} = \type{\poly{S^{\prime\prime}},\poly{\role{\kappa}}}$.
    Now we observe the context reduction using \cref{l:sr4},\cref{l:sr5}, \cref{lem:red-commutes-subt} and \cref{rule:g-com2}:
    \begin{flalign}
        &\Gammab{\o}, \Gammab{!} \redC \Gammapp \subT \Gammab{0},\{\ses{s}[\q] \hasT \type{S_L} + \ses{s}[\p] \hasT \type{!\role{\alpha} \&_{i\in I } m_i (\poly{T_i^\prime}) . S_{R_i} \| S_{R_k}\sub{\q}{\role{\alpha}}\sub{\poly{\role{r}}}{\poly{\role{\kappa}}}}\}
        && \reason{by \cref{rule:g-com2}}{sr6} \\
        &\Gammapp = \left(\type{\Gamma_L^\prime} \subT \Gammab{0},\ses{s}[\q] \hasT \type{S_L} \right) \cSplit \Gammab{!} \cSplit \left(\type{\Gamma_R^\prime} \subT \Gammab{0},\ses{s}[\p] \hasT \type{S_{R_k}\sub{\q}{\role{\alpha}}\sub{\poly{\role{r}}}{\poly{\role{\kappa}}}}\right)
        && \reason{by \cref{l:sr6} and \cref{fig:context-operations}}{sr7}
    \end{flalign}
    where $\Gammab{0}$ contains any singletons in $\Gammab{\o}, \Gammab{!}$, and therefore, $\pEnd(\Gammab{0})$.

    Using rule \cref{rule:g-cong}, we can infer the shape of the entire context reduction:
    \begin{flalign}
        & \Gamma \redC \Gammap = (\Gammab{L_0} \cSplit \Gammab{\role \rho} (\cSplit \Gammab{L_i})_{i \in 1..n} \cSplit \Gammab{R_0}) + (\Gammapp)
        && \reason{by \cref{l:sr6},\cref{l:sr7} and \cref{rule:g-cong}}{sr8} \\
        & \Gammap = (\Gammab{R_0} \cSplit \Gammab{!}) \cSplit (\Gammab{L_0} \cSplit \type{\Gamma_L^\prime}) \cSplit (\type{\Gamma_R^\prime} (\cSplit \Gammab{L_i})_{i \in 1..n})
        && \reason{by \cref{l:sr7} and \cref{l:sr8}}{sr9}
    \end{flalign}
    where $\Gammab{\role \rho}$ is engulfed in $\type{\Gamma_R^\prime}$ since $\type{\Gamma_R^\prime}$ contains $\Gammab{0}$, which contains all singletons that exist in $\Gammab{\o},\Gammab{\&}$.

    We can now observe that $\Gammap \vdash P^\prime$ since:
    \begin{flalign}
        & \Gammab{R_0} \cSplit \Gammab{!} \vdash {!}\rcvI{\ses{s}[\role{p}]}{\role{\alpha}}{\m_i(\poly{b_i})} Q_i^\prime
        && \reason{from der. tree}{sr10} \\
        & \Gammab{L_0} \cSplit \type{\Gamma_L^\prime} \vdash Q
        && \reason{from der. tree, \cref{l:sr7} and \cref{lem:narrowing-proc}}{sr11} \\
        & \type{\Gamma_R^\prime} (\cSplit \Gammab{L_i})_{i \in 1..n} \vdash Q_k^\prime \sub{\poly{d}}{\poly{b_k}}\sub{\q}{\role \alpha}
        && \reason{from der. tree, \cref{l:sr7}, \cref{lem:narrowing-proc}, \cref{lem:role-subs-proc} and \cref{lem:chan-subs}}{sr12} \\
        & \Gammap \vdash P^\prime && \reason{by \cref{l:sr10},\cref{l:sr11},\cref{l:sr12}}{sr13}
    \end{flalign}

    Lastly, we must show that $\Gammap$ is \emph{safe}, which holds from \cref{ass:sr2}, \cref{rule:safe-r} and \cref{l:sr8}. \qed

\end{proof}

\begin{customthm}{2}
    Assume $\emptyset\cons\Gamma \vdash P$ with $\Gamma$ safe and $P \equiv
    {\bigm|_{\q \in I}\! P_\q}$ where each $P_\q$ is either $\0$ (up-to
    $\equiv$), or only plays role $\role q$ in $\ses{s}$. 
    Then, $\Gamma \redC$ implies $\exists \Gammap,P^\prime$ s.t. $\Gamma \redC \Gammap$, $P \redP^* P^\prime$ and $\Gammap \vdash P^\prime$, where $P^\prime \equiv {\bigm|_{\q \in I}\! P_\q^\prime}$ and each $P_\q^\prime$ is either $\0$ (up-to $\equiv$), or only plays role $\role q$ in $\ses{s}$.
\end{customthm}

\begin{proof}[Sketch]
    We start by observing that if $\varphi(\Gamma)$ and $\Gamma\redC$, then we can split out of $\Gamma$ either 
    \emph{(i)} a selection and branch type; or \emph{(ii)} a selection and replicated branch type.
    Proceeding in a similar manner for both cases, we infer the shape of $P$ based on the split types.
    It follows that $P$ can reduce via \Cref{rule:r-c} (or \Cref{rule:r-bang1}/\Cref{rule:r-bang2} in case \emph{(ii)}), possibly after some applications of \Cref{rule:r-ctx},\Cref{rule:r-proc-call}, and \Cref{rule:r-choice}; and $\Gamma$ can reduce such that the process reduction remains typable.
    Lastly, we show that the process reductum retains the structure of the assumptions for session fidelity (by rule \Cref{rule:r-cong} and since \Cref{def:one-role} requires its properties to hold on all subterms). 
    We note that for replicated communication, a process $P_\p$ can take the shape $P_\p \equiv {\bigm|_{i \in I}\! P_i^\prime}$ where each $P_i^\prime$ is either $\0$ (up-to $\equiv$), or only plays role $\role p$ in $\ses{s}$; \ie\ a process playing a single role can consist of multiple parallel processes that all play the same role. 
    In this case we first show that $P_\p$ is typed under a runtime type
    $\type{U}$, then proceed in a similar manner as before.
\end{proof}

%% file: decidability.tex
\section{Unfolding}\label{app:dec}

Below we present our definitions for type and context unfolding; they are mostly standard, we only adapt type unfolding to cater for parallel types.

\begin{definition}[Type unfolding]
    The one-step unfolding of a type $\type{U}$, written \textnormal{$\unf(\type{U})$}, is given by:
    {\normalfont\[
        \unf(\type{\mu \recVar{t}.S}) = \St\sub{\type{\mu \recVar{t}.S}}{\type{\recVar{t}}}
        \tab\tab
        \unf(\type{S \| U}) = \unf(\St) \type{\|} \unf(\type{U})
        \tab\tab
        \unf(\St) = \St \tab \textit{if } \St \neq \type{\mu \recVar{t}.S^\prime}
    \]}%
    The n-steps unfolding of a type $\type{U}$, written \textnormal{$\unf^n(\type{U})$}, is given by:
    {\normalfont\[
        \unf^0(\type{U}) = \type{U}
        \tab\tab
        \unf^{m+1}(\type{U}) = \unf(\unf^m(\type{U}))
    \]}%
    The complete unfolding of a type $\type{U}$, written $\unf^*(\type{U})$, is defined as:
    {\normalfont\[
        \unf^*(\type{U}) = \unf^n(\type{U}) \tab \textit{for the smallest $n$ s.t. } \unf^n(\type{U}) = \unf^{n+1}(\type{U})
    \]}
\end{definition}

\begin{definition}[Context unfolding, as in \cite{DBLP:journals/pacmpl/ScalasY19tech} definition K.1]
    The set of unfoldings of a type context $\Gamma$, written \textnormal{$\unf(\Gamma)$}, is defined below (where $\Gamma\sub{\type{U}}{c}$ is a mapping update):
    {\normalfont\[
        \unf(\Gamma) = \textstyle\bigcup_{c \hasT \type{U} \in \Gamma}\{\Gamma\sub{\unf(\type{U})}{c}\}
        \tab\tab
        \textit{extends to sets of contexts as}
        \tab\tab
        \unf(\{\Gammab{i}\}_{i \in I}) = \textstyle\bigcup_{i \in I}\unf(\Gammab{i})
    \]}%
    Given a set of contexts $\xi$, the closure of its unfoldings, written \textnormal{$\unf*(\xi)$}, is defined as:
    {\normalfont\[
        \unf^*(\xi) = \textsf{lfix}(\lambda \xi^\prime . \xi \cup \xi^\prime \cup \unf(\xi \cup \xi^\prime)) \tab \textit{where }\textsf{lfix}\textit{ is the least fixed point of its argument}
    \] }
\end{definition}

%% file: main.bbl
\begin{thebibliography}{10}
\providecommand{\url}[1]{\texttt{#1}}
\providecommand{\urlprefix}{URL }
\providecommand{\doi}[1]{https://doi.org/#1}

\bibitem{DBLP:journals/iandc/AlmeidaMTV22}
Almeida, B., Mordido, A., Thiemann, P., Vasconcelos, V.T.: Polymorphic lambda
  calculus with context-free session types. Inf. Comput.  \textbf{289}(Part),
  104948 (2022). \doi{10.1016/j.ic.2022.104948},
  \url{https://doi.org/10.1016/j.ic.2022.104948}

\bibitem{DBLP:conf/concur/BarwellSY022}
Barwell, A.D., Scalas, A., Yoshida, N., Zhou, F.: Generalised multiparty
  session types with crash-stop failures. In: Klin, B., Lasota, S., Muscholl,
  A. (eds.) 33rd International Conference on Concurrency Theory, {CONCUR} 2022,
  September 12-16, 2022, Warsaw, Poland. LIPIcs, vol.~243, pp. 35:1--35:25.
  Schloss Dagstuhl - Leibniz-Zentrum f{\"{u}}r Informatik (2022).
  \doi{10.4230/LIPIcs.CONCUR.2022.35},
  \url{https://doi.org/10.4230/LIPIcs.CONCUR.2022.35}

\bibitem{DBLP:conf/concur/BettiniCDLDY08}
Bettini, L., Coppo, M., D'Antoni, L., Luca, M.D., Dezani{-}Ciancaglini, M.,
  Yoshida, N.: Global progress in dynamically interleaved multiparty sessions.
  In: van Breugel, F., Chechik, M. (eds.) {CONCUR} 2008 - Concurrency Theory,
  19th International Conference, {CONCUR} 2008, Toronto, Canada, August 19-22,
  2008. Proceedings. Lecture Notes in Computer Science, vol.~5201, pp.
  418--433. Springer (2008). \doi{10.1007/978-3-540-85361-9\_33},
  \url{https://doi.org/10.1007/978-3-540-85361-9\_33}

\bibitem{DBLP:conf/forte/BrunD24}
Brun, M.A.L., Dardha, O.: Mag{\(\pi\)}!: The role of replication in typing
  failure-prone communication. In: Castiglioni, V., Francalanza, A. (eds.)
  Formal Techniques for Distributed Objects, Components, and Systems - 44th
  {IFIP} {WG} 6.1 International Conference, {FORTE} 2024, Held as Part of the
  19th International Federated Conference on Distributed Computing Techniques,
  DisCoTec 2024, Groningen, The Netherlands, June 17-21, 2024, Proceedings.
  Lecture Notes in Computer Science, vol. 14678, pp. 99--117. Springer (2024).
  \doi{10.1007/978-3-031-62645-6\_6},
  \url{https://doi.org/10.1007/978-3-031-62645-6\_6}

\bibitem{DBLP:conf/esop/CairesP17}
Caires, L., P{\'{e}}rez, J.A.: Linearity, control effects, and behavioral
  types. In: Yang, H. (ed.) Programming Languages and Systems - 26th European
  Symposium on Programming, {ESOP} 2017, Held as Part of the European Joint
  Conferences on Theory and Practice of Software, {ETAPS} 2017, Uppsala,
  Sweden, April 22-29, 2017, Proceedings. Lecture Notes in Computer Science,
  vol. 10201, pp. 229--259. Springer (2017).
  \doi{10.1007/978-3-662-54434-1\_9},
  \url{https://doi.org/10.1007/978-3-662-54434-1\_9}

\bibitem{DBLP:journals/mscs/CairesPT16}
Caires, L., Pfenning, F., Toninho, B.: Linear logic propositions as session
  types. Math. Struct. Comput. Sci.  \textbf{26}(3),  367--423 (2016).
  \doi{10.1017/S0960129514000218},
  \url{https://doi.org/10.1017/S0960129514000218}

\bibitem{DBLP:conf/sfm/CoppoDPY15}
Coppo, M., Dezani{-}Ciancaglini, M., Padovani, L., Yoshida, N.: A gentle
  introduction to multiparty asynchronous session types. In: Bernardo, M.,
  Johnsen, E.B. (eds.) Formal Methods for Multicore Programming - 15th
  International School on Formal Methods for the Design of Computer,
  Communication, and Software Systems, {SFM} 2015, Bertinoro, Italy, June
  15-19, 2015, Advanced Lectures. Lecture Notes in Computer Science, vol.~9104,
  pp. 146--178. Springer (2015). \doi{10.1007/978-3-319-18941-3\_4},
  \url{https://doi.org/10.1007/978-3-319-18941-3\_4}

\bibitem{DBLP:conf/fossacs/DardhaG18}
Dardha, O., Gay, S.J.: A new linear logic for deadlock-free session-typed
  processes. In: Baier, C., Lago, U.D. (eds.) Foundations of Software Science
  and Computation Structures - 21st International Conference, {FOSSACS} 2018,
  Held as Part of the European Joint Conferences on Theory and Practice of
  Software, {ETAPS} 2018, Thessaloniki, Greece, April 14-20, 2018, Proceedings.
  Lecture Notes in Computer Science, vol. 10803, pp. 91--109. Springer (2018).
  \doi{10.1007/978-3-319-89366-2\_5},
  \url{https://doi.org/10.1007/978-3-319-89366-2\_5}

\bibitem{DBLP:journals/iandc/DardhaGS17}
Dardha, O., Giachino, E., Sangiorgi, D.: Session types revisited. Inf. Comput.
  \textbf{256},  253--286 (2017). \doi{10.1016/j.ic.2017.06.002},
  \url{https://doi.org/10.1016/j.ic.2017.06.002}

\bibitem{DBLP:journals/corr/abs-1208-6483}
Deni{\'{e}}lou, P., Yoshida, N., Bejleri, A., Hu, R.: Parameterised multiparty
  session types. Log. Methods Comput. Sci.  \textbf{8}(4) (2012).
  \doi{10.2168/LMCS-8(4:6)2012}, \url{https://doi.org/10.2168/LMCS-8(4:6)2012}

\bibitem{Ehrhard18}
Ehrhard, T.: An introduction to differential linear logic: proof-nets, models
  and antiderivatives. Math. Struct. Comput. Sci.  \textbf{28}(7),  995--1060
  (2018)

\bibitem{FowlerKDLM23}
Fowler, S., Kokke, W., Dardha, O., Lindley, S., Morris, J.G.: Separating
  sessions smoothly. Log. Methods Comput. Sci.  \textbf{19}(3) (2023)

\bibitem{DBLP:journals/acta/GayH05}
Gay, S.J., Hole, M.: Subtyping for session types in the pi calculus. Acta
  Informatica  \textbf{42}(2-3),  191--225 (2005).
  \doi{10.1007/S00236-005-0177-Z},
  \url{https://doi.org/10.1007/s00236-005-0177-z}

\bibitem{DBLP:conf/concur/Honda93}
Honda, K.: Types for dyadic interaction. In: Best, E. (ed.) {CONCUR} '93, 4th
  International Conference on Concurrency Theory, Hildesheim, Germany, August
  23-26, 1993, Proceedings. Lecture Notes in Computer Science, vol.~715, pp.
  509--523. Springer (1993). \doi{10.1007/3-540-57208-2\_35},
  \url{https://doi.org/10.1007/3-540-57208-2\_35}

\bibitem{DBLP:conf/esop/HondaVK98}
Honda, K., Vasconcelos, V.T., Kubo, M.: Language primitives and type discipline
  for structured communication-based programming. In: Hankin, C. (ed.)
  Programming Languages and Systems - ESOP'98, 7th European Symposium on
  Programming, Held as Part of the European Joint Conferences on the Theory and
  Practice of Software, ETAPS'98, Lisbon, Portugal, March 28 - April 4, 1998,
  Proceedings. Lecture Notes in Computer Science, vol.~1381, pp. 122--138.
  Springer (1998). \doi{10.1007/BFb0053567},
  \url{https://doi.org/10.1007/BFb0053567}

\bibitem{DBLP:journals/jacm/HondaYC16}
Honda, K., Yoshida, N., Carbone, M.: Multiparty asynchronous session types. J.
  {ACM}  \textbf{63}(1),  9:1--9:67 (2016). \doi{10.1145/2827695},
  \url{https://doi.org/10.1145/2827695}

\bibitem{KokkeMP19}
Kokke, W., Montesi, F., Peressotti, M.: Better late than never: a
  fully-abstract semantics for classical processes. Proc. {ACM} Program. Lang.
  \textbf{3}({POPL}),  24:1--24:29 (2019)

\bibitem{KokkeMW20}
Kokke, W., Morris, J.G., Wadler, P.: Towards races in linear logic. Log.
  Methods Comput. Sci.  \textbf{16}(4) (2020)

\bibitem{DBLP:conf/esop/BrunD23}
{Le Brun}, M.A., Dardha, O.: Mag{\(\pi\)}: Types for failure-prone
  communication. In: Wies, T. (ed.) Programming Languages and Systems - 32nd
  European Symposium on Programming, {ESOP} 2023, Held as Part of the European
  Joint Conferences on Theory and Practice of Software, {ETAPS} 2023, Paris,
  France, April 22-27, 2023, Proceedings. Lecture Notes in Computer Science,
  vol. 13990, pp. 363--391. Springer (2023).
  \doi{10.1007/978-3-031-30044-8\_14},
  \url{https://doi.org/10.1007/978-3-031-30044-8\_14}

\bibitem{DBLP:journals/corr/abs-2203-12875}
Marshall, D., Orchard, D.: Replicate, reuse, repeat: Capturing non-linear
  communication via session types and graded modal types. In: Carbone, M.,
  Neykova, R. (eds.) Proceedings of the 13th International Workshop on
  Programming Language Approaches to Concurrency and Communication-cEntric
  Software, PLACES@ETAPS 2022, Munich, Germany, 3rd April 2022. {EPTCS},
  vol.~356, pp. 1--11 (2022). \doi{10.4204/EPTCS.356.1},
  \url{https://doi.org/10.4204/EPTCS.356.1}

\bibitem{DBLP:conf/esop/PocasCMV23}
Po{\c{c}}as, D., Costa, D., Mordido, A., Vasconcelos, V.T.: System
  f${}^\mu_{\omega}$ with context-free session types. In: Wies, T. (ed.)
  Programming Languages and Systems - 32nd European Symposium on Programming,
  {ESOP} 2023, Held as Part of the European Joint Conferences on Theory and
  Practice of Software, {ETAPS} 2023, Paris, France, April 22-27, 2023,
  Proceedings. Lecture Notes in Computer Science, vol. 13990, pp. 392--420.
  Springer (2023). \doi{10.1007/978-3-031-30044-8\_15},
  \url{https://doi.org/10.1007/978-3-031-30044-8\_15}

\bibitem{DBLP:journals/pacmpl/QianKB21}
Qian, Z., Kavvos, G.A., Birkedal, L.: Client-server sessions in linear logic.
  Proc. {ACM} Program. Lang.  \textbf{5}({ICFP}),  1--31 (2021).
  \doi{10.1145/3473567}, \url{https://doi.org/10.1145/3473567}

\bibitem{RochaC23}
Rocha, P., Caires, L.: Safe session-based concurrency with shared linear state.
  In: {ESOP}. Lecture Notes in Computer Science, vol. 13990, pp. 421--450.
  Springer (2023)

\bibitem{DBLP:journals/darts/ScalasDHY17}
Scalas, A., Dardha, O., Hu, R., Yoshida, N.: A linear decomposition of
  multiparty sessions for safe distributed programming (artifact). Dagstuhl
  Artifacts Ser.  \textbf{3}(2),  03:1--03:2 (2017). \doi{10.4230/DARTS.3.2.3},
  \url{https://doi.org/10.4230/DARTS.3.2.3}

\bibitem{DBLP:journals/pacmpl/ScalasY19tech}
Scalas, A., Yoshida, N.: Less is more: multiparty session types revisited.
  Tech. Rep.~6, Imperial College London (2018),
  \url{https://www.doc.ic.ac.uk/research/technicalreports/2018/6}

\bibitem{DBLP:journals/pacmpl/ScalasY19}
Scalas, A., Yoshida, N.: Less is more: multiparty session types revisited.
  Proc. {ACM} Program. Lang.  \textbf{3}({POPL}),  30:1--30:29 (2019).
  \doi{10.1145/3290343}, \url{https://doi.org/10.1145/3290343}

\bibitem{DBLP:phd/dnb/Stutz24}
Stutz, F.: Implementability of Asynchronous Communication Protocols - The Power
  of Choice. Ph.D. thesis, Kaiserslautern University of Technology, Germany
  (2024), \url{https://kluedo.ub.rptu.de/frontdoor/index/index/docId/8077}

\bibitem{DBLP:conf/icfp/ThiemannV16}
Thiemann, P., Vasconcelos, V.T.: Context-free session types. In: Garrigue, J.,
  Keller, G., Sumii, E. (eds.) Proceedings of the 21st {ACM} {SIGPLAN}
  International Conference on Functional Programming, {ICFP} 2016, Nara, Japan,
  September 18-22, 2016. pp. 462--475. {ACM} (2016).
  \doi{10.1145/2951913.2951926}, \url{https://doi.org/10.1145/2951913.2951926}

\bibitem{DBLP:journals/toplas/ToninhoY18}
Toninho, B., Yoshida, N.: Interconnectability of session-based logical
  processes. {ACM} Trans. Program. Lang. Syst.  \textbf{40}(4),  17:1--17:42
  (2018)

\bibitem{DBLP:journals/iandc/Vasconcelos12}
Vasconcelos, V.T.: Fundamentals of session types. Inf. Comput.  \textbf{217},
  52--70 (2012). \doi{10.1016/J.IC.2012.05.002},
  \url{https://doi.org/10.1016/j.ic.2012.05.002}

\bibitem{DBLP:journals/jfp/Wadler14}
Wadler, P.: Propositions as sessions. J. Funct. Program.  \textbf{24}(2-3),
  384--418 (2014). \doi{10.1017/S095679681400001X},
  \url{https://doi.org/10.1017/S095679681400001X}

\end{thebibliography}
